\documentclass[twocolumn]{article}
\usepackage{amssymb}

\usepackage{xurl}
\usepackage[super,comma,sort&compress]{natbib}
\usepackage{abstract}
\usepackage{lipsum}
\usepackage{etex}
\usepackage{amsmath}
\usepackage{graphicx}
\usepackage{hyperref}

\newcommand{\abstractText}{\noindent
We show that graph-theoretic problem CLIQUE
can't be solved in polynomial time by any deterministic TM. This upgrades
the well-known \emph{monotone} unsolvability
and eventually implies  $\textbf{P}\neq \textbf{NP}$ as
CLIQUE is NP-complete. As compared to familiar Razborov-style arguments 
 we use Boolean logic coupled with standard circuit formalism.
}

\newtheorem{theorem}{Theorem}

\newtheorem{corollary}[theorem]{Corollary}

\newtheorem{definition}[theorem]{Definition}

\newtheorem{lemma}[theorem]{Lemma}

\newtheorem{partial solution}[theorem]{Partial Solution}

\newenvironment{proof}[1][Proof]{\textbf{#1.} }{\ \rule{0.5em}{0.5em}}

\input{tcilatex}
\hypersetup{colorlinks=true, urlcolor=blue, linkcolor=blue, citecolor=blue}

\begin{document}

\title{On P Versus NP}
\author{L. Gordeev}


\twocolumn[
  \begin{@twocolumnfalse}
    \maketitle
    \begin{abstract}
      \abstractText
      \newline
      \newline
    \end{abstract}
  \end{@twocolumnfalse}
]


\section{\protect\large Introduction}

Our proof is based on the following two observations, where CLIQUE$_{m,k}$
says that a given graph on $\leq m$ vertices has a clique of $k$ vertices.

1. Computational complexity of Boolean circuits is polynomial in the
complexity of \emph{De\thinspace Morgan normal} (abbr.: DMN) $(\vee ,\wedge
) $-circuits $C^{\pm }$ allowing negated inputs (literals) $\lnot v_{i}$
along with variables $v_{i}$.

2. For sufficiently large natural numbers $m=k^{4}$, the size of any given
DMN circuit solution $C^{\pm }$ of CLIQUE$_{m,k}$ is exponential in $m$. To
prove this claim we consider disjoint pairs of graphs on $\leq m$ vertices
(called double graphs below) $D=\left\langle G\left( \varepsilon \right)
,G^{\prime }\right\rangle $, for chosen assignments $\varepsilon :\left[ n%
\right] \rightarrow \left\{ 0,1\right\} $, such that 
\begin{equation*}
\begin{array}{l}
G\left( \varepsilon \right) =\left\{ \pi (i):i\in \left[ n\right]
\&\,\varepsilon \left( i\right) =1\right\} \ \text{and\newline
} \\ 
G^{\prime }\subseteq \left\{ \pi (i):i\in \left[ n\right] \&\,\varepsilon
\left( i\right) =0\right\}
\end{array}
\end{equation*}
where $n:=\binom{m}{2}$ is the total number of edges involved and $\pi $ is
standard 1--1 enumeration thereof. We'll write $D^{+}$ and $D^{-}$ for $%
G\left( \varepsilon \right) $ and $G^{\prime } $, respectively, and denote
by $\mathcal{D}$ the set of all double graphs assuming that every $D\in 
\mathcal{D}$ is represented by a DMN circuit

\begin{equation*}
C_{D}^{\pm }:=\!\underset{\pi (i)\in D^{+}}{\bigwedge }\!v_{i}\wedge \!\!%
\underset{\pi (i)\in D^{-}}{\bigwedge }\lnot v_{i}.
\end{equation*}

For any $\varepsilon :\left[ n\right] \rightarrow \left\{ 0,1\right\} $ we
denote by $C^{\pm }\left( \varepsilon \right) $ a variable-free Boolean
circuit defined by replacing all inputs $v_{i}$ and $\lnot v_{i}$ occurring
in $C^{\pm }$ by $\varepsilon \left( i\right) $ and $1\!-\!\varepsilon
\left( i\right) $, respectively, and let $\left\| C^{\pm }\left( \varepsilon
\right) \right\| \in \left\{ 0,1\right\} $ be Boolean value of $C^{\pm
}\left( \varepsilon \right) $.

We say that CLIQUE$_{m,k}$ ($k<m$) is \emph{decidable} by $C^{\pm }$ iff 
\begin{equation*}
\forall \varepsilon \left( \left\| C^{\pm }\left( \varepsilon \right)
\right\| =1\Leftrightarrow CLIQ\left( \varepsilon \right) \right) \quad
\qquad \qquad \qquad \ \left( 1\right)
\end{equation*}
where 
\begin{equation*}
CLIQ\left( \varepsilon \right) :\Leftrightarrow ``G\left( \varepsilon
\right) \ has\ a\ clique\ of\ k\ vertices".
\end{equation*}

Now consider any Boolean circuit $C$ with gates $\vee $, $\wedge $ and/or $%
\lnot $, whose inputs are supplied with Boolean constants and variables $%
v_{i}$ ($i\in \left[ n\right] $). Corresponding DNM circuit $C^{\pm }$
arises by applying to gates of $C$ the well-known De\thinspace Morgan
rewriting rules 1--4:

\begin{enumerate}
\item  $\lnot \,1\hookrightarrow 0 \, ,\ \lnot \,0 \hookrightarrow 1 .$

\item  $\lnot \left( a\vee b\right) \hookrightarrow \lnot \,a\wedge \lnot \,
b.$

\item  $\lnot \left( a\wedge b\right) \hookrightarrow \lnot \,a\vee \lnot \,
b.$

\item  $\lnot \lnot \,a\hookrightarrow a.$
\end{enumerate}

It is a folklore that circuit size of $C^{\pm }$\ is at most that of $C$.
Namely, circuit structure of $C^{\pm }$ arises by successively replacing $%
\lnot \,1$ by $0$, $\lnot \,0$ by $1$ and gates $\vee $ and $\wedge $
occurring in the scope of $\lnot $ in $C$ by the complementary $\wedge $ and 
$\vee $, respectively, while deleting gates $\lnot $ in question and all
double negations $\lnot \lnot $. For any chosen $\varepsilon :\left[ n\right]
\rightarrow \left\{ 0,1\right\} $, let $C\left( \varepsilon \right) $ be a
variable-free Boolean circuit obtained by substituting $\varepsilon \left(
i\right) \in \left\{ 0,1\right\} $ for every input $v_{i}$ from $C$. Let $%
\left\| C\left( \varepsilon \right) \right\| \in \left\{ 0,1\right\} $
denote Boolean value of $C\left( \varepsilon \right) $. Then clearly 
\begin{equation*}
\forall \varepsilon \left( \left\| C\left( \varepsilon \right) \right\|
=\left\| C^{\pm }\left( \varepsilon \right) \right\| \right) \qquad \qquad
\qquad \qquad \qquad \,\left( 2\right)
\end{equation*}
Moreover, assuming that CLIQUE$_{m,k}$ is decidable by $C$ we suppose that 
\begin{equation*}
\forall \varepsilon \left( \left\| C\left( \varepsilon \right) \right\|
=1\Leftrightarrow CLIQ\left( \varepsilon \right) \right) \qquad \qquad
\qquad \quad \quad \ \left( 3\right)
\end{equation*}
Furthermore, we define a DNF expansion $\mathrm{DN}\left( C^{\pm }\right)
\subset \mathcal{D}$ such that 
\begin{equation*}
\forall \varepsilon \left( 
\begin{array}{c}
\left\| C^{\pm }\left( \varepsilon \right) \right\| =1\Leftrightarrow \\ 
\left( \exists D\in \mathrm{DN}\left( C^{\pm }\right) \right) \left\|
C_{D}^{\pm }\left( \varepsilon \right) \right\| =1\smallskip
\end{array}
\right) \quad \qquad \ \ \ \left( 4\right)
\end{equation*}

Having done this we show that the hypothetical equivalence 
\begin{equation*}
\forall \varepsilon \left( 
\begin{array}{c}
CLIQ\left( \varepsilon \right) \Leftrightarrow \\ 
\left( \exists D\in \mathrm{DN}\left( C^{\pm }\right) \right) \left\|
C_{D}^{\pm }\left( \varepsilon \right) \right\| =1\smallskip
\end{array}
\right) \quad \qquad \ \ \ \left( 5\right)
\end{equation*}
implies that the size of $C^{\pm }$, and hence that of $C$, is exponential
in $m$, which eventually implies $\mathbf{P}\neq \mathbf{NP}$ (in fact, $%
\mathbf{NP}\subseteq \mathbf{P/poly}$). It remains to observe that $\left(
5\right) $ follows from conjunction of $\left( 1\right) $, $\left( 2\right) $%
, $\left( 3\right) $, $\left( 4\right) $.

That (5) \textbf{implies exponential complexity of} $C^{\pm }$ is our
crucial technical refinement of the known ``monotone'' approach. To prove
this implication we adapt Razborov-style method of approximations and
estimate resulting deviations, a.k.a. ``error sets'', via modified Erd\H{o}%
s-Rado lemma and double sunflowers with quasi-positive cores $\left\langle
G,\emptyset \right\rangle $. The sets of double graphs $\mathrm{DN}\left(
C^{\pm }\right) $ occurring in (5) are double-graph analogues of plain \emph{%
crude circuits} considered in the ``monotone'' proofs.

The whole proof is exposed in Chapters 1, 2 below. For obvious reasons it is
more technical than standard presentations of the monotone approach, which
omit logic formalism of the notions involved.

\subparagraph{Acknowledgment}

I would like to thank Ren\'{e} Thiemann who took the time to verify crucial
proofs with the theorem prover Isabelle \cite{R}, whose work was extremely
helpful in finding flaws and errors in the earlier presentations.

\section{\protect\large Preliminaries}

\subsection{\protect\large Basic notations}

\begin{itemize}
\item  In the sequel we assume

\fbox{$2<\ell <p<k\leq m^{\frac{1}{4}}\ $and$\ L=\left( p-1\right) ^{\ell
}\ell !$}

\item  For any $A,B\subseteq \left[ m\right] $ let

\fbox{$A\ast B:=\left\{ \left\{ x,y\right\} :x\in A\ \&\ y\in B\ \&\ x\neq
y\right\} $} \newline
and \fbox{$A^{\left( 2\right) }:=A\ast A$}, where $\left[ m\right] :=\left\{
1,\cdots ,m\right\} $.

So $\left| \left[ m\right] ^{\left( 2\right) }\right| =\binom{m}{2}=\frac{1}{%
2}m\left( m-1\right) $, where $\left| S\right| :=card\left( S\right) $.

\item  For any $X\subseteq \left[ m\right] ^{\left( 2\right) }$ let

\fbox{\textsc{v}$\left( X\right) :=\left\{ x\in \left[ m\right] :\left(
\exists y\in \left[ m\right] \right) \left\{ x,y\right\} \in X\right\} $}

and \fbox{$\wp _{0}X:=\left\{ Y:Y\subset X\right\} $} (the proper subsets of 
$X$).

\item  Let \fbox{$\mathcal{F}:=\left\{ f:\left[ m\right] \rightarrow \left[
k-1\right] \right\} $}\ and for any $f\in \mathcal{F}$

\fbox{$C_{f}:=\left\{ \left\{ x,y\right\} \in \left[ \mathrm{Dom}\left(
f\right) \right] ^{\left( 2\right) }:f\left( x\right) \neq f\left( y\right)
\right\} $}
\end{itemize}

\subsection{\protect\large Plain and double graphs}

\begin{itemize}
\item  Call $\mathcal{G}:=\wp _{0}\left[ m\right] ^{\left( 2\right) }$ the
set of graphs (unordered, possibly empty) on (at most) $m$ vertices.
For any $\emptyset \neq G\in \mathcal{G}$ call pairs $\left\{ x,y\right\}
\in \mathcal{G}$ and v$\left( G\right) $ the edges and vertices,
respectively.

\item  \fbox{$\mathrm{POS}:=\mathcal{K\!}:=\!\left\{ \text{\textsc{%
v\negthinspace }}\left( G\right) ^{\left( 2\right) }:\left| \text{\textsc{v}}%
\left( G\right) \right| =k\right\} $} and

$\fbox{$\mathrm{CLIQ}:=\left\{ G\in \mathcal{G}:\left( \exists K\in \mathcal{%
K}\right) K\subseteq G\right\} $}$

are called positive tests and plain clique problem, respectively.

\item  \fbox{$\mathrm{NEG}:=\left\{ C_{f}:f\in \mathcal{F}\right\} $} and

$\fbox{$\mathrm{ACLIQ}:=\left\{ G\in \mathcal{G}:\left( \exists H\in \mathrm{%
NEG}\right) G\subseteq H\right\} $}$

are called\smallskip\ negative tests and plain anticliques, respectively$.$

\item  Pairs of disjoint plain graphs are called double graphs. That is,

\fbox{$\mathcal{D}:=\left\{ 
\begin{array}{c}
\left\langle G,H\right\rangle \in \mathcal{G\times G}: \\ 
G\cap H=\emptyset \ \&\ G\cup H\in \wp _{0}\left[ m\right] ^{\left( 2\right)
}
\end{array}
\right\} $}\smallskip\ is the set of double graphs. Double graph $%
\left\langle \emptyset ,\emptyset \right\rangle $ is identified with $%
\emptyset $, while $\mathcal{G}$ regarded part of $\mathcal{D}$ via $G\ni 
\mathcal{G}\hookrightarrow \left\langle G,\emptyset \right\rangle \in 
\mathcal{D}$

and/or $G\ni \mathcal{G}\hookrightarrow \left\langle \emptyset
,G\right\rangle \in \mathcal{D}$.

\item  For any $D=\left\langle G,H\right\rangle \in \mathcal{D}$ and $%
\mathcal{X}\subseteq \mathcal{D}$ let

\fbox{$D^{+}\!:=\!G$, $D^{-}\!:=\!H$}$\,\in \mathcal{G}$ and

\fbox{$\mathcal{X}^{+}:=\left\{ D^{+}:D\in \mathcal{X}\right\} $}$%
\,\subseteq \mathcal{G}$ and

\fbox{$\mathcal{X}^{-}:=\left\{ D^{-}:D\in \mathcal{X}\right\} $}$%
\,\subseteq \mathcal{G}$.

\item  For any $G\in \mathcal{G}$ and $D\in \mathcal{D}$ let

\fbox{$G\subseteq ^{+}D:\Leftrightarrow G\subseteq \!D^{+}\!$}.

\item  $\fbox{$\mathrm{CLIQ}_{2}:=\left\{ D\in \mathcal{D}:\left( \exists
K\in \mathrm{POS}\right) K\subset ^{+}D\right\} $}$ are called double
cliques.

\item  $\fbox{$\mathrm{ACLIQ}_{2}:=\left\{ D\in \mathcal{D}:\left( \exists
G\in \mathrm{NEG}\right) G\subseteq ^{+}D\right\} $}$ are called double
anticliques.
\end{itemize}

\begin{lemma}
$\mathrm{CLIQ}_{2}\!\cap \!\mathrm{ACLIQ}$\thinspace $_{2}=\emptyset $.
\negthinspace Moreover $\left| \mathrm{POS}\right| $ $\!=\!\left( \QATOP{m}{k%
}\right) $ and $\left| \mathcal{F}\right| \!=\!\left( k-1\right)
^{m}\!>\left| \mathrm{NEG}\right| .$
\end{lemma}

\begin{proof}
This is easily verified (precise estimation of $\left| \mathrm{NEG}\right| $
is unimportant).
\end{proof}

\subsection{\protect\large Basic operations on double graphs}

Except for standard set-theoretic operations $\cup $ and $\cap $\ we
consider double union $\Cup :\mathcal{D}\times \mathcal{D}\rightarrow 
\mathcal{D}$ and double product $\odot :\wp \mathcal{D}\times \wp \mathcal{D}%
\rightarrow \wp \mathcal{D}$.

\begin{itemize}
\item  For any $D,E\in \mathcal{D}$ let

$D\Cup E:=\!\left\{ \!\! 
\begin{array}{ccc}
\left\langle D^{+}\!\cup \!E^{+}\!,D^{-}\!\cup \!E^{-}\right\rangle ,\! & 
\!if\ it & \!\!\!\!\in \mathcal{D}, \\ 
\emptyset , & \!else. & 
\end{array}
\!\right. $
\end{itemize}

\begin{itemize}
\item  For any $\mathcal{X},\mathcal{Y}\subseteq \mathcal{D}$ let

\fbox{$\mathcal{X\odot Y}:=\!\left\{ D\Cup E:\left\langle D,E\right\rangle
\in \mathcal{X}\times \mathcal{Y}\right\} $}$\,\subseteq \!\mathcal{D}.$
\end{itemize}

Note that $\emptyset \odot \mathcal{Y}=\mathcal{X}\odot \emptyset =\emptyset 
$. The following conditions easily hold for any $\mathcal{X,Y,X^{\prime
},Y^{\prime }}\subseteq \mathcal{D}.$

\begin{enumerate}
\item  $\mathcal{X}\odot \mathcal{Y}=\mathcal{Y}\odot \mathcal{X},\ \mathcal{%
X}\odot \left( \mathcal{Y\odot Z}\right) =\left( \mathcal{X\odot Y}\right)
\odot \mathcal{Z}.$

\item  $\mathcal{X}\odot \left( \mathcal{Y\cup Z}\right) =\left( \mathcal{%
X\odot Y}\right) \cup \left( \mathcal{X\odot Z}\right) ,\ \mathcal{X}\cup
\left( \mathcal{Y\odot Z}\right) \subseteq \left( \mathcal{X\cup Y}\right)
\odot \left( \mathcal{X\cup Z}\right) .$

\item  $\mathcal{X}\subseteq\mathcal{X}^{\prime }\&\,\mathcal{Y}\subseteq 
\mathcal{Y}^{\prime }\Rightarrow \mathcal{X}\odot\mathcal{Y}\subseteq 
\mathcal{X}^{\prime }\odot \mathcal{Y}^{\prime }.$
\end{enumerate}

\section{\protect\large Proof proper}

\subsection{\protect\large Acceptability}

With any given set of double graphs $\mathcal{X}$ we correlate accepted
double tests $\mathrm{AC}\left( \mathcal{X}\right) \subseteq D,$ $\mathrm{AC}%
^{\text{\textsc{p}}}\left( \mathcal{X}\right) \subseteq \mathrm{POS}$ and
negative double coloring $\mathrm{AC}^{\text{\textsc{n}}}\left( \mathcal{X}%
\right) \subseteq \mathrm{NEG}$. $\smallskip $Corresponding sets of accepted
double tests, resp. colorings, are as follows.

\begin{definition}
For any $\mathcal{X}\subseteq \mathcal{D}$ and $D\in \mathcal{D}$ let $%
\mathcal{X\vdash }D$ abbreviate $\left( \exists E\in \mathcal{X}\right)
E^{+}\subseteq D\smallskip ^{+}$. Then let:

\begin{enumerate}
\item  $\mathrm{AC}\left( \mathcal{X}\right) :=\left\{ D\in \mathcal{D}:%
\mathcal{X\vdash }D\right\} $.

\item  $\mathrm{AC}^{\text{\textsc{p}}}\left( \mathcal{X}\right) :=\mathrm{AC%
}\left( \mathcal{X}\right) ^{+}\cap \mathrm{POS}$,

\item  $\mathrm{AC}^{\text{\textsc{n}}}\left( \mathcal{X}\right) :=\mathrm{AC%
}\left( \mathcal{X}\right) ^{+}\cap \mathrm{NEG}$.
\end{enumerate}
\end{definition}

\begin{lemma}
Conditions 1--6 hold for any $\mathcal{X},\mathcal{Y}\subseteq \mathcal{D}$.

\begin{enumerate}
\item  $\mathrm{AC}\left( \mathcal{\emptyset }\right) =\mathrm{AC}^{\text{%
\textsc{p}}}\left( \mathcal{\emptyset }\right) =\mathrm{AC}^{\text{\textsc{n}%
}}\left( \mathcal{\emptyset }\right) =\emptyset $.

\item  $\mathrm{AC}\left( \mathcal{D}\right) =\mathcal{D},\ \mathrm{AC}^{%
\text{\textsc{p}}}\left( \mathrm{POS}\right) =\mathrm{POS},$

\item  $\mathrm{AC}^{\text{\textsc{n}}}\left( \mathrm{NEG}\right) =\mathrm{%
NEG}.$

\item  If $X\subseteq Y$ then $\mathrm{AC}^{\left( \text{\textsc{p}}\right)
\left( \text{\textsc{n}}\right) }\left( \mathcal{X}\right) \subseteq \mathrm{%
AC}^{\left( \text{\textsc{p}}\right) \left( \text{\textsc{n}}\right) }\left( 
\mathcal{Y}\right) $.

\item 
\begin{enumerate}
\item  $\mathrm{AC}\left( \mathcal{X\cup Y}\right) =\mathrm{AC}\left( 
\mathcal{X}\right) \cup \mathrm{AC}\left( \mathcal{Y}\right) $,

\item  $\mathrm{AC}^{\text{\textsc{p}}}\left( \mathcal{X\cup Y}\right) =%
\mathrm{AC}^{\text{\textsc{p}}}\left( \mathcal{X}\right) \cup \mathrm{AC}^{%
\text{\textsc{p}}}\left( \mathcal{Y}\right) $,

\item  $\mathrm{AC}^{\text{\textsc{n}}}\left( \mathcal{X\cup Y}\right) =%
\mathrm{AC}^{\text{\textsc{n}}}\left( \mathcal{X}\right) \cup \mathrm{AC}^{%
\text{\textsc{n}}}\left( \mathcal{Y}\right) $.
\end{enumerate}

\item 
\begin{enumerate}
\item  $\mathrm{AC}\left( \mathcal{X\cap Y}\right) \!\subseteq \!\mathrm{AC}%
\left( \mathcal{X\odot Y}\right) \!=\newline
\mathrm{AC}\left( \mathcal{X}\right) \cap \mathrm{AC}\left( \mathcal{Y}%
\right) $,

\item  $\mathrm{AC}^{\text{\textsc{p}}}\left( \mathcal{X\cap Y}\right)
\!\subseteq \!\mathrm{AC}^{\text{\textsc{p}}}\left( \mathcal{X\odot Y}%
\right) \!=\newline
\mathrm{AC}^{\text{\textsc{p}}}\left( \mathcal{X}\right) \cap \mathrm{AC}^{%
\text{\textsc{p}}}\left( \mathcal{Y}\right) $,

\item  $\mathrm{AC}^{\text{\textsc{n}}}\left( \mathcal{X\cap Y}\right)
\!\subseteq \!\mathrm{AC}^{\text{\textsc{n}}}\left( \mathcal{X\odot Y}%
\right) \!=\newline
\mathrm{AC}^{\text{\textsc{n}}}\left( \mathcal{X}\right) \cap \mathrm{AC}^{%
\text{\textsc{n}}}\left( \mathcal{Y}\right) $.
\end{enumerate}
\end{enumerate}
\end{lemma}

\begin{proof}
1--5: trivial.

6 (a).\ It will suffice to prove $\mathrm{AC}\left( \mathcal{X\odot Y}%
\right) =\mathrm{AC}\left( \mathcal{X}\right) \cap \mathrm{AC}\left( 
\mathcal{Y}\right) $. So suppose $D\in \mathrm{AC}\left( \mathcal{X\odot Y}%
\right) $, i.e. $\mathcal{X\odot Y\vdash }D$, i.e.$\ $there are $E_{1}\in 
\mathcal{X}$ and $E_{2}\in \mathcal{Y}$ such that $E_{1}^{+}\cup
E_{2}^{+}\subseteq D^{+}$, which by 
\begin{equation*}
E_{1}^{+}\cup E_{2}^{+}\subseteq D^{+}\Leftrightarrow E_{1}^{+}\subseteq
D^{+}\,\&\,\,E_{2}^{+}\subseteq D^{+}
\end{equation*}
yields both $D\in \mathrm{AC}\left( \mathcal{X}\right) $ and $D\in \mathrm{AC%
}\left( \mathcal{Y}\right) $. Suppose $D\in \mathrm{AC}\left( \mathcal{X}%
\right) \cap \mathrm{AC}\left( \mathcal{Y}\right) $, i.e. $\mathcal{X\vdash }%
D$ and $\mathcal{Y\vdash }D$, i.e. there are $E_{1}\in \mathcal{X}$ and $%
E_{2}\in \mathcal{Y}$ such that $E_{1}^{+}\subseteq D^{+}$ and $%
E_{2}^{+}\subseteq D^{+}$, and hence $E_{1}\Cup E_{2}\in \mathcal{X\odot Y} $%
, which by the same token yields $D\in \mathrm{AC}\left( \mathcal{X\odot Y}%
\right) $.

6 (b), (c) follow analogously.
\end{proof}

\subsection{\protect\large Approximations and deviations}

In what follows we generalize conventional monotone approach, cf. e.g. \cite
{AA}, \cite{AR1}, \cite{AR2}, \cite{AR3}, \cite{P}, \cite{S}, \cite{J}, \cite
{Y}, \cite{R}, \cite{ER}. We supply operations\ $\cup $ and $\odot $ on $\wp 
\mathcal{D}$ with their approximators $\sqcup $ and $\sqcap $ operating on
arbitrary subsets $\mathcal{X\subseteq D}$ such that for all $D$ from $%
\mathcal{X}$, $\left| \text{\textsc{v}}\left( D^{+}\right) \right| \leq \ell 
$ (thus we approximate only positive parts of double graphs). We define
corresponding deviations $\partial _{\sqcup }^{\text{\textsc{p}}},\partial
_{\sqcup }^{\text{\textsc{n}}},$ $\partial _{\sqcap }^{\text{\textsc{p}}%
},\partial _{\sqcap }^{\text{\textsc{n}}}$\ from $\cup $ and $\odot $ with
respect to accepted test graphs and show that these deviations make
``small'' fractions thereof (Lemmata 10, 11). These deviations are analogous
to ``error sets'' caused by approximations in conventional monotone approach
based on the Erd\H{o}s-Rado lemma \cite{ER} (cf. \cite{P}, \cite{J}, \cite{Y}%
).

\subsubsection{\protect\large Basic notations and definitions}

\begin{itemize}
\item  Let $\mathcal{G}^{\ell }:=\left\{ G\in \mathcal{G}:\left| \text{%
\textsc{v}}\left( G\right) \right| \leq \ell \right\} $ and

$\mathcal{D}^{\ell }:=\left\{ D\in \mathcal{D}:\left| \text{\textsc{v}}%
\left( D^{+}\right) \right| \leq \ell \right\} $. \footnote{{\small {%
{\footnotesize Note that $G\in G^{\ell }$ implies}}}
\par
{\small $\sqrt{2\left| G\right| }<\frac{1}{2}\left( 1+\sqrt{1+8\left|
G\right| }\right) \leq \left| \text{\textsc{v}}\left( G\right) \right| \leq
2\ell .$}}

Let $\wp _{L}\mathcal{D}:=\left\{ \mathcal{X\subseteq D}:\left| \mathcal{X}%
^{+}\right| \leq L\right\} $ and

$\wp _{L}\mathcal{D}^{\ell }$ $:=$ $\left\{ \mathcal{X\subseteq D}^{\ell
}:\left| \mathcal{X}^{+}\right| \leq L\right\} $.

\item  If $D,E\in \mathcal{D}^{\ell }$ and $\mathcal{X,Y}\subseteq \mathcal{D%
}^{\ell }$, let

$D\,\Cup ^{\ell }E:=\left\{ 
\begin{array}{ccc}
D\Cup E, & if & it\in \mathcal{D}^{\ell }, \\ 
\emptyset & else, & 
\end{array}
\right. $

\item  $\mathcal{X}\odot ^{\ell }\mathcal{Y}:=\left\{ D\Cup ^{\ell }\!E\in 
\mathcal{D}^{\ell }:D\in \mathcal{X}\ \&\ E\in \mathcal{Y}\right\} \in \wp \,%
\mathcal{D}^{\ell }$.

\item  Together with double graphs consider double sets

$\mathcal{S}=\left\{ \left\langle A,B\right\rangle :A,B\subseteq \left[ m%
\right] \ \&\ A\cap B=\emptyset \right\} $, where for $S=\left\langle
A,B\right\rangle \in \mathcal{S}$ we denote $A$ and $B$ by $S^{+}$ and $%
S^{-} $, respectively.

Let $\mathcal{S}^{\ell }:=\!\left\{ S\!\in \!\mathcal{S}:\left| S^{+}\right|
\leq \ell \right\} $,

$\wp _{L}\mathcal{S}:=\left\{ \mathcal{X\subseteq S}:\left| \mathcal{X}%
^{+}\right| \leq L\right\} $ and

$\wp _{L}\mathcal{S}^{\ell }$ := $\left\{ \mathcal{X\subseteq S}^{\ell
}:\left| \mathcal{X}^{+}\right| \leq L\right\} $, where

$\mathcal{X}^{+}=\left\{ S^{+}:S\in \mathcal{X}\right\} .$

\item  For any $G,H\in G$ and $D\in D$ we let

\textsc{v}$\left( G\setminus \!\!\!\setminus H\right) :=\QTR{sc}{v}\left(
G\right) \setminus \QTR{sc}{v}\left( H\right) $ and

\textsc{s}$\left( D\right) :=\,\left\langle \text{\textsc{v}}\left(
D^{+}\setminus \!\!\!\setminus D^{-}\right) ,\text{\textsc{v}}\left(
D^{-}\setminus \!\!\!\setminus D^{+}\right) \right\rangle \in \mathcal{S}$.

For any $\mathcal{X}\subseteq \mathcal{\!D}$ let \textsc{s}$\left( \mathcal{X%
}\right) :=\left\{ \text{\textsc{s}}\left( D\right) :D\in \mathcal{X}%
\right\} \!\subseteq \!\mathcal{S}$.

Then for any $X\mathcal{\!}\subseteq \!\mathcal{D}$, $\mathcal{Y}\subseteq 
\mathcal{D}^{\ell }$ and $\mathcal{Z}\in \wp _{L}\mathcal{D}^{\ell }$

we get \textsc{s}$\left( \mathcal{X}\right) \subseteq \mathcal{S}$, \textsc{s%
}$\left( \mathcal{Y}\right) \subseteq \mathcal{S}^{\ell }$ and \textsc{s}$%
\left( \mathcal{Z}\right) \in \wp _{L}\mathcal{S}^{\ell }$.
\end{itemize}

\begin{definition}
A collection of double sets $\mathcal{V}=\left\{ V_{1},\cdots ,V_{p}\right\}
\subset $ $\mathcal{S}$ is called a sunflower with $p$ (different) positive
petals $V_{1}^{+},\cdots ,V_{p}^{+}$ if $V_{1}^{+}\cap
V_{2}^{+}=V_{i}^{+}\cap V_{j}^{+}$\ holds for all $i<j\in \left[ p\right] $.
Then $V_{\copyright }:=\!\left\langle V_{1}^{+}\cap V_{2}^{+},\emptyset
\right\rangle $ is called the core of $\mathcal{V}$.
\end{definition}

\begin{lemma}
Any given $\mathcal{U}\subseteq \mathcal{S}^{\ell }$ such that $\left| \, 
\mathcal{U}^{+}\right| >L$ contains a sunflower $\mathcal{V}\subset \mathcal{%
U}$ with $p$ positive petals $V_{1}^{+},\cdots ,V_{p}^{+}$ and core $%
V_{\copyright }\in \mathcal{S}^{\ell }$.
\end{lemma}

\begin{proof}
By the original Erd\H{o}s-Rado lemma \cite{ER}, since $L=\left( p-1\right)
^{\ell }\ell !$.
\end{proof}

\begin{definition}[plucking]
Plucking algorithm 
\begin{equation*}
\wp \mathcal{D}^{\ell }\ni \mathcal{Z}\mapsto \mathrm{PL}\left( \mathcal{Z}%
\right) \in \wp _{L}\mathcal{D}^{\ell }
\end{equation*}
arises by recursion on $\left| \text{\textsc{s}}\left( \mathcal{Z}\right)
\right| $. If $\left| \text{\textsc{s}}\left( \mathcal{Z}\right) ^{+}\right|
\leq L$, let $\mathrm{PL}\left( \mathcal{Z}\right) :=\mathcal{Z}$.
Otherwise, let $\mathcal{Z}_{0}:=\mathcal{Z}$, thus $\left| \text{\textsc{s}}%
\left( \mathcal{Z}_{0}\right) ^{+}\right| >L$. By the last lemma with
respect to $U:=\text{\textsc{s}}\left( \mathcal{Z}_{0}\right) \subseteq 
\emph{S}^{\ell }$ we choose a sunflower of cardinality $p$, $\mathcal{V}%
=\left\{ V_{1},\cdots ,V_{p}\right\} \subseteq \,$\textsc{s}$\left( \mathcal{%
Z}_{0}\right) $ with positive petals $V_{1}^{+},\cdots ,V_{p}^{+}$ and core $%
V_{\copyright }=\!\left\langle V_{1}^{+}\cap V_{2}^{+},\emptyset
\right\rangle \in \mathcal{S}^{\ell }$. Let $\mathcal{Z}_{0}^{\prime
}:=\left\{ D\in \mathcal{Z}_{0}:\left( \exists j\in \left[ p\right] \right) 
\text{\textsc{s}}\left( D\right) =V_{j}\right\} \subseteq \mathcal{D}^{\ell
} $ and $D_{\copyright }=\left\langle D_{\copyright }^{+},\emptyset
\right\rangle \in \mathcal{D}^{\ell }$, where $D_{\copyright }^{+}:=\bigcap
\left\{ D^{+}\!:D\in \mathcal{Z}_{0}^{\prime }\right\} $, which also yields 
\textsc{s}$\left( D_{\copyright }\right) ^{+}\subseteq V_{\copyright }^{+}$.
Then rewrite $\mathcal{Z}_{0}$ to $\mathcal{Z}_{1}$ that arises by replacing
every $D\in \mathcal{Z}_{0}^{\prime }$ by $D_{\copyright }$. \footnote{%
{\small {\footnotesize This operation will be referred to as elementary
plucking.}}} Note that $\left| \text{\textsc{s}}\left( \mathcal{Z}%
_{1}\right) ^{+}\right| \leq \left| \text{\textsc{s}}\left( \mathcal{Z}%
_{0}\right) ^{+}\right| -p+1.$ If $\left| \text{\textsc{s}}\left( \mathcal{Z}%
_{1}\right) ^{+}\right| \leq L$, let $\mathrm{PL}\left( \mathcal{Z}\right) :=%
\mathcal{Z}_{1}$. Otherwise, if $\left| \text{\textsc{s}}\left( \mathcal{Z}%
_{1}\right) ^{+}\right| >L$, we analogously pass from $\mathcal{Z}%
_{1}\subseteq D^{\ell }$ to $\mathcal{Z}_{2}\subseteq D^{\ell }$. Proceeding
this way we eventually arrive at $\mathcal{Z}_{q}\subseteq \mathcal{D}^{\ell
}$ with $\left| \text{\textsc{s}}\left( \mathcal{Z}_{q}\right) ^{+}\right|
\leq L$ and then let $\mathrm{PL}\left( \mathcal{Z}\right) :=Z_{q}$.
\end{definition}

\begin{lemma}
For any given $\mathcal{Z}\in \wp \mathcal{D}^{\ell }$, $\mathrm{PL}\left( 
\mathcal{Z}\right) \in \wp _{L}\mathcal{D}^{\ell }$ requires $<\left| \text{%
\textsc{s}}\left( \mathcal{Z}\right) ^{+}\right| $ elementary pluckings.
That is, if $\mathrm{PL}\left( \mathcal{Z}\right) :=Z_{q}$ as above, then $%
q<\left| \text{\textsc{s}}\left( \mathcal{Z}\right) ^{+}\right| $.
\end{lemma}

\begin{proof}
Each elementary plucking reduces the number of sets at least by $p-1$. Hence

$q<\left| \text{\textsc{s}}\left( \mathcal{Z}\right) ^{+}\right| \left(
p-1\right) ^{-1}$ $<\left| \text{\textsc{s}}\left( \mathcal{Z}\right)
^{+}\right| $.
\end{proof}

\begin{definition}
For any $\mathcal{X},\mathcal{Y}\in \wp \mathcal{D}^{\ell }$ call the
following operations $\sqcup $, $\sqcap $ and sets $\mathcal{X}\sqcup 
\mathcal{Y}$, $\mathcal{X}\sqcap \mathcal{Y}$ the approximators and
approximations of operations $\cup $, $\odot $ and sets $\mathcal{X}\cup 
\mathcal{Y}$, $\mathcal{X}\odot \mathcal{Y}$, respectively, which determine\
deviations $\partial _{\sqcup }^{\text{\textsc{p}}},\partial _{\sqcup }^{%
\text{\textsc{n}}},\partial _{\sqcap }^{\text{\textsc{p}}},\partial _{\sqcap
}^{\text{\textsc{n}}}$ with respect to the accepted tests. \footnote{{\small 
{\footnotesize We write $\partial $ instead of $\delta $ used in \cite{AR1}}}%
$-${\small {\footnotesize \cite{AR3}.}}}

\begin{enumerate}
\item  $\mathcal{X}\sqcup \mathcal{Y}:=\mathrm{PL}\left( \mathcal{X\cup Y}%
\right) \in \wp _{L}\mathcal{D}^{\ell }$.

\item  $\mathcal{X}\sqcap \mathcal{Y}:=\mathrm{PL}\left( \mathcal{X\odot
^{\ell }Y}\right) \in \wp _{L}\mathcal{D}^{\ell }$.

\item  $\partial _{\sqcup }^{\text{\textsc{p}}}\left( \mathcal{X,Y}\right) :=%
\mathrm{AC}^{\text{\textsc{p}}}\!\left( \mathcal{X\cup Y}\right) \setminus 
\mathrm{AC}^{\text{\textsc{p}}}\!\left( \mathcal{X\sqcup Y}\right) \subseteq 
\mathrm{POS}.$

\item  $\partial _{\sqcap }^{\text{\textsc{p}}}\left( \mathcal{X,Y}\right) :=%
\mathrm{AC}^{\text{\textsc{p}}}\!\left( \mathcal{X\odot Y}\right) \setminus 
\mathrm{AC}^{\text{\textsc{p}}}\!\left( \mathcal{X\sqcap Y}\right) \subseteq 
\mathrm{POS}.$

\item  $\partial _{\sqcup }^{\text{\textsc{n}}}\left( \mathcal{X,Y}\right) =%
\mathrm{AC}^{\text{\textsc{n}}}\!\left( \mathcal{X\sqcup Y}\right) \setminus 
\mathrm{AC}^{\text{\textsc{n}}}\!\left( \mathcal{X\cup Y}\right) \subseteq 
\mathrm{NEG}.$

\item  $\partial _{\sqcap }^{\text{\textsc{n}}}\left( \mathcal{X,Y}\right) :=%
\mathrm{AC}^{\text{\textsc{n}}}\!\left( \mathcal{X\sqcap Y}\right) \setminus 
\mathrm{AC}^{\text{\textsc{n}}}\!\!\left( \mathcal{X\odot Y}\right)
\subseteq \mathrm{NEG}.$
\end{enumerate}

For$\mathcal{U}\subseteq \mathrm{NEG}$ let $\left| \mathcal{U}\right| ^{\ast
}:=\left| \left\{ f\in \mathcal{F}:C_{f}\subseteq \mathcal{U}\right\}
\right| $ (functional cardinality of $U$). In particular $\left| \mathrm{NEG}%
\right| ^{\ast }=\mathcal{F}$. In the sequel we use functional cardinality
as our basic measure of the number of negative double tests involved.
\end{definition}

\subsubsection{\protect\large Upper bounds}

We assume that $m$ is sufficiently large and $k=2\ell^{2}$.

\begin{lemma}
For any $D\in \mathcal{D}^{\ell }$ let

$R_{\subseteq }\left( D\right) :=\left\{ f\in \mathcal{F}:D^{+}\subseteq
C_{f}\right\} $ and

$R_{\nsubseteq }\left( D\right) \!:=\!\left\{ f\!\in \!\mathcal{F}%
\!:\!D^{+}\!\nsubseteq \!C_{f}\right\} =\mathcal{F}\setminus \!R_{\subseteq
}\!\left( D\right) $. Then $\left| \mathcal{R}_{\subseteq }\left( D\right)
\right| \!\geq \!\dfrac{1}{2}\left| \mathcal{F}\right| $ and $\left| 
\mathcal{R}_{\nsubseteq }\left( D\right) \right| \!\leq \!\dfrac{1}{2}\left| 
\mathcal{F}\right| $. Moreover, for any $D_{1},\cdots ,D_{q}\in D^{\ell }$
such that $\left( \forall i\neq j\in \left[ q\right] \right) D_{i}^{+}\cap
D_{j}^{+}=\emptyset $ it holds $\left| \overset{q}{\underset{i=1}{\bigcap }}%
\mathcal{R}_{\nsubseteq }\left( D_{i}\right) \right| \leq 2^{-q}\left| 
\mathcal{F}\right| $.
\end{lemma}

\begin{proof}
For any $G\in \mathcal{G}^{\ell }$ let $R_{\subseteq }\left( G\right)
:=\left\{ f\in \mathcal{F}:\,G\subseteq C_{f}\right\} $. This yields by
standard monotone arguments $\left| \mathcal{R}_{\subseteq }\left( G\right)
\right| \geq \dfrac{1}{2}\left| \mathcal{F}\right| $, which for any $D\in 
\mathcal{D}^{\ell }$ implies $\left| \mathcal{R}_{\subseteq }\left( D\right)
\right| \geq \dfrac{1}{2}\left| \mathcal{F}\right| $ and hence $\left| 
\mathcal{R}_{\nsubseteq }\left( D\right) \right| =\left| \mathcal{F}%
\setminus \mathcal{R}_{\subseteq }\left( D\right) \right| \leq \dfrac{1}{2}%
\left| \mathcal{F}\right| $ (see Appendix A). To establish the last
assertion it will suffice to observe that for any $j\in \left[ q-1\right] $,
it holds 
\begin{equation*}
\mathbb{P}\left[ \overset{q}{\underset{i=j}{\bigcap }}\mathcal{R}%
_{\nsubseteq }\left( D_{i}\right) \right] =\mathbb{P}\left[ \mathcal{R}%
_{\nsubseteq }\left( D_{j}\right) \right] \cdot \mathbb{P}\left[ \overset{q}{%
\underset{i=j+1}{\bigcap }}\mathcal{R}_{\nsubseteq }\left( D_{i}\right) %
\right]
\end{equation*}
where for any $\mathcal{X}\subseteq \mathcal{F}$ we set $\mathbb{P} \left[ 
\mathcal{X}\right] :=\left| \mathcal{X}\right| \left| \mathcal{F}\right|
^{-1}$ (the probability). The latter holds by standard arguments, as $%
R_{\nsubseteq }\left( D_{1}\right) ,\cdots ,R_{\nsubseteq
}\left(D_{q}\right) $ are independent events in the space $\mathcal{F}$ (see
also Appendix A).
\end{proof}

\begin{lemma}
Let $\mathcal{Z}=\mathcal{X}\cup \mathcal{Y}\in \wp \mathcal{D}^{\ell }$ and 
$\mathrm{PL}\left( \mathcal{Z}\right) \in \wp _{L}\mathcal{D}^{\ell }$ for $%
\mathcal{X,Y}\in \wp _{L}\mathcal{D}^{\ell }$. 
Then $\mathrm{PL}\left( \mathcal{Z}\right) $ requires $<2L$ elementary
pluckings. Moreover \fbox{$\partial _{\sqcup }^{\text{\textsc{p}}}\left( 
\mathcal{X,Y}\right) =0$} while $\fbox{$\left| \partial _{\sqcup }^{\text{%
\textsc{n}}}\left( \mathcal{X,Y}\right) \right| ^{\ast }<2^{1-p}L\left| 
\mathcal{F}\right| $}$.
\end{lemma}

\begin{proof}
We argue as in the analogous monotone case using Lemmata 7, 9. Let $\mathcal{%
V}=\left\{ V_{1},\cdots ,V_{p}\right\} $ $\subseteq \,$\text{\textsc{s}}$%
\left( \left( \mathcal{X\cup Y}\right) _{i}\right) $ be the sunflower with
positive petals $V_{1},\cdots ,V_{p}$ and core $V_{\copyright
}=\!\left\langle V_{1}^{+}\cap V_{2}^{+},\emptyset \right\rangle \in 
\mathcal{S}^{\ell }$ arising at $i^{th}$ elementary plucking ($i>0$) and let 
$D_{\copyright }=\left\langle D_{\copyright }^{+},\emptyset \right\rangle
\in D^{\ell }$ be the corresponding double graph. Consider $\mathrm{PL}%
\left( \mathcal{Z}\right) $ and corresponding $\partial _{\sqcup }^{\text{%
\textsc{p}}}\left( \mathcal{X,Y}\right) $ and $\partial _{\sqcup }^{\text{%
\textsc{n}}}\left( \mathcal{X,Y}\right) $. $\partial _{\sqcup }^{\text{%
\textsc{p}}}\left( \mathcal{X,Y}\right) =\emptyset $ is clear as elementary
pluckings replace some (plain) graphs by subgraphs and thereby preserve the
accepted positive tests.\newline
Now consider $\partial _{\sqcup }^{\text{\textsc{n}}}\left( \mathcal{X,Y}%
\right) $. We estimate the total number of fake negative double tests that
arise after rewriting $\mathcal{Z}_{i-1}\hookrightarrow \mathcal{Z}_{i}$
involved. Suppose $\mathcal{Z}_{i}$ is obtained by substituting $%
D_{\copyright }=\left\langle D_{\copyright }^{+},\emptyset \right\rangle
=\left\langle \bigcap \left\{ D^{+}\!:D\in \mathcal{Z}_{0}^{\prime }\right\}
,\emptyset \right\rangle \in D^{\ell }$, for every $D\in Z_{i-1}^{\prime }$,
where $Z_{i-1}^{\prime }=\left\{ D\in \mathcal{Z}_{i-1}:\left( \exists j\in %
\left[ p\right] \right) \text{\textsc{s}}\left( D\right) =V_{j}\right\} $.
Let $\left| \mathcal{Z}_{i-1}^{\prime }\right| =p^{\prime }\geq p$ with $%
\mathcal{Z}_{i-1}^{\prime }=\left\{ D_{1},\cdots ,D_{p^{\prime }}\right\} $.
Now let $C_{f}\in \mathrm{NEG}$ be any fake negative test created by this
substitution. I.e. $D_{\copyright }^{+}\subseteq C_{f}$, although for every $%
t\in \left[ p^{\prime }\right] $, we have $D_{t}^{+}\nsubseteq C_{f}$. Let $%
D_{t}^{\prime }:=\left\langle D_{t}^{+}\setminus D_{\copyright
}^{+},D_{t}^{-}\right\rangle \in D^{\ell }$. Note that for any $s\neq t\in %
\left[ p^{\prime }\right] $ we have $D_{s}^{\prime +}\cap D_{t}^{\prime
+}=\emptyset \neq D_{t}^{+}$, while $\,$\textsc{s}$\left( D_{\copyright
}\right) $ contains the only common nodes of $D_{s}^{+}$ and $D_{t}^{+}$.
Furthermore by Lemma 9 we know that $\mathbb{P}\left[ \mathcal{R}%
_{\nsubseteq }\left( D_{t}^{\prime }\right) \right] \leq \mathbb{P}\left[ 
\mathcal{R}_{\nsubseteq }\left( D_{t}\right) \right] \leq \dfrac{1}{2}$
holds for every $t\in \left[ p^{\prime }\right] $. Summing up, by Lemma 9 we
obtain 
\begin{eqnarray*}
&&\mathbb{P}\left[ \underset{t=1}{\overset{p^{\prime }}{\bigcap }}\mathcal{R}%
_{\nsubseteq }\left( D_{t}\right) \cap \mathcal{R}_{\subseteq }\left(
D_{\copyright }\right) \right] \\
&=&\mathbb{P}\left[ \underset{t=1}{\overset{p^{\prime }}{\bigcap }}\left( 
\mathcal{R}_{\nsubseteq }\left( D_{t}\right) \cap \mathcal{R}_{\subseteq
}\left( D_{\copyright }\right) \right) \right] \\
&\leq &\mathbb{P}\left[ \underset{t=1}{\overset{p^{\prime }}{\bigcap }}%
\mathcal{R}_{\nsubseteq }\left( D_{t}^{\prime }\right) \right] \leq
2^{-p^{\prime }} \\
&\leq &2^{-p}
\end{eqnarray*}
Hence with regard to functional cardinality there are less than 
\begin{eqnarray*}
\left| \underset{t=1}{\overset{p^{\prime }}{\bigcap }}\mathcal{R}%
_{\nsubseteq }\left( D_{t}\right) \cap \mathcal{R}_{\subseteq }\left(
D_{\copyright }\right) \right| &=& \\
\mathbb{P}\left[ \underset{t=1}{\overset{p^{\prime }}{\bigcap }}\mathcal{R}%
_{\nsubseteq }\left( D_{t}\right) \cap \mathcal{R}_{\subseteq }\left(
D_{\copyright }\right) \right] \left| \mathcal{F}\right| &\leq &2^{-p}\left| 
\mathcal{F}\right|
\end{eqnarray*}
fake negative tests $C_{f}$ created by the replacement $\mathcal{Z}%
_{i-1}\hookrightarrow \mathcal{Z}_{i}$. Recall that by Lemma 7\ there are $%
q<L$ elementary pluckings involved. This yields 
\begin{eqnarray*}
\partial _{\sqcup }^{\text{\textsc{n}}}\left( \mathcal{X,Y}\right)
&\subseteq &\overset{q-1}{\underset{i=0}{\bigcup }}\partial _{\sqcup }^{%
\text{\textsc{n}}}\left( \mathcal{X,Y}\right) _{i}\text{ for } \\
\partial _{\sqcup }^{\text{\textsc{n}}}\left( \mathcal{X,Y}\right) _{i} &:&=%
\mathrm{AC}^{\text{\textsc{n}}}\!\left( \mathcal{X\cup Y}\right)
_{i+1}\setminus \mathrm{AC}^{\text{\textsc{n}}}\!\left( \mathcal{X\cup Y}%
\right) _{i}.
\end{eqnarray*}
Hence $\left| \partial _{\sqcup }^{\text{\textsc{n}}}\left( \mathcal{X,Y}%
\right) \right| ^{\ast }\leq \overset{q-1}{\underset{i=0}{\sum }}\left|
\partial _{\sqcup }^{\text{\textsc{n}}}\left( \mathcal{X,Y}\right)
_{i}\right| ^{\ast }<q2^{-p}\left| \mathcal{F}\right| <2^{1-p}L\left| 
\mathcal{F}\right| $.
\end{proof}

\begin{lemma}
Let $\mathcal{X,Y}\in \wp _{L}\mathcal{D}^{\ell }$, $\mathcal{X}\odot ^{\ell
}\mathcal{Y}\in \wp \,\mathcal{D}^{\ell }$ and $\mathcal{Z}=\mathrm{PL}%
\left( \mathcal{X\odot }^{\ell }\mathcal{Y}\right) \in \wp _{L}\mathcal{D}%
^{\ell }$. Thus $\left| \text{\textsc{s}}\left( \mathcal{Z}\right)^{+}
\right| \leq L$ and $\left| \text{\textsc{s}}\left( \mathcal{X\mathcal{\odot 
}Y}\right)^{+} \right| \leq L^{2}.$ Then $\fbox{$\left| \partial _{\sqcap }^{%
\text{\textsc{p}}}\left( \mathcal{X,Y}\right) \right| <L^{2}\left( \QATOP{%
m-\ell -1}{k-\ell -1}\right) $}$ and $\fbox{$\left| \partial _{\sqcap }^{%
\text{\textsc{n}}}\left( \mathcal{X,Y}\right) \right| ^{\ast
}<2^{-p}L^{2}\left| \mathcal{F}\right| $}$.
\end{lemma}

\begin{proof}
$\left| \partial _{\sqcap }^{\text{\textsc{n}}}\left( \mathcal{X,Y}\right)
\right| \!^{\ast }<2^{-p}L^{2}\left| \mathcal{F}\right| $ is analogous to
the inequality for $\partial _{\sqcup }^{\text{\textsc{n}}}\left( \mathcal{%
X,Y}\right) $. Consider $\partial _{\sqcap }^{\text{\textsc{p}}}\left( 
\mathcal{X,Y}\right) $. We adapt standard arguments used in the ``monotone''
proofs (cf. e.g. \cite{P}, \cite{Y}). It is readily seen that deviations can
only arise by deleting a $D\cup E\notin \mathcal{D}^{\ell }$ for some $%
D,E\in \mathcal{D}^{\ell }$ when passing from $\mathcal{X}\odot \mathcal{Y}$
to $\mathcal{X}\odot ^{\ell }\mathcal{Y}$ (note that $\mathcal{X}\odot 
\mathcal{Y}$ can completely disappear, in which case $\mathrm{PL}\left( 
\mathcal{X\odot }^{\ell }\mathcal{Y}\right) =\mathcal{X\odot }^{\ell }%
\mathcal{Y}=\emptyset $). So suppose $H\in \left( \mathcal{X\mathcal{\odot }Y%
}\right) \setminus \mathcal{D}^{\ell }$. Thus $\ell <\left| \text{\textsc{v}}%
\left( H^{+}\right) \right| \leq 2\ell $. Let us estimate $\left| \mathcal{K}%
_{H}\right| $\ for $K_{H}:=\left\{ K\in \mathrm{POS}:H^{+}\subseteq
K\right\} $. Note that $\ell <\left| \text{\textsc{v}}\left( H^{+}\right)
\right| $\ implies that $K_{H}$ contains at most $\left( \QATOP{m-\ell -1}{%
k-\ell -1}\right) $ cliques $K$. So $\left| \mathcal{K}_{H}\right| \leq
\left( \QATOP{m-\ell -1}{k-\ell -1}\right) $ and 
\begin{eqnarray*}
\partial _{\sqcap }^{\text{\textsc{p}}}\left( \mathcal{X,Y}\right)
&\subseteq &\bigcup \left\{ \mathcal{K}_{H}:H\in \mathcal{\left( \mathcal{X%
\mathcal{\odot }Y}\right) \setminus D^{\ell }}\right\} \\
&\subseteq &\bigcup \left\{ \mathcal{K}_{H}:H\in \mathcal{\mathcal{X\mathcal{%
\odot }Y}}\right\},
\end{eqnarray*}
which by $\left| \text{\textsc{s}}\left( \mathcal{X\mathcal{\odot }Y}\right)
^{+}\right| \leq L^{2}$ and Lemma 7 yields the result.
\end{proof}

\subsection{\protect\large Formalism}

We'll formalize previous considerations in basic De\thinspace Morgan logic
with atomic negation (called DMN logic) over $\left( \QATOP{m}{2}\right) $
distinct variables. For any given DMN formula $\varphi $ we define its
double graph representation $\mathrm{DN}\left( \varphi \right) $ and
approximation $\mathrm{AP}\left( \varphi \right) $ augmented with total
deviations $\partial ^{\text{\textsc{p}}}\left( \varphi \right) \subseteq 
\mathrm{POS}$ and $\partial ^{\text{\textsc{n}}}\left( \varphi \right)
\subseteq \mathrm{NEG}$. Using our estimates on $\partial _{\sqcup }^{\text{%
\textsc{p}}},\partial _{\sqcup }^{\text{\textsc{n}}},\partial _{\sqcap }^{%
\text{\textsc{p}}},\partial _{\sqcap }^{\text{\textsc{n}}}$ we show that $%
\mathrm{AC}^{\text{\textsc{p}}}\left( \mathrm{DN}\left( \varphi \right)
\right) =\mathrm{POS}$ plus $\mathrm{AC}^{\text{\textsc{n}}}\left( \mathrm{DN%
}\left( \varphi \right) \right) =\emptyset $ infers exponential circuit size
of $\varphi $ (cf. Theorem 14 below).

\subsubsection{\protect\large Syntax}

In the sequel we let $n:=\left( \QATOP{m}{2}\right) =\frac{1}{2}m\left(
m-1\right) $ and $\pi :\left[ n\right] \overset{1-1}{\longrightarrow }\left[
m\right] ^{\left( 2\right) }$.

\begin{itemize}
\item  Let $\mathcal{A}$ denote Boolean algebra with constants $0 ,1$,
operations $\vee ,\wedge $, atomic negation $\lnot $ and variables $v_{i}$
for any $i\in \left[ n\right] $. That is, formulas (abbr.: $\varphi ,\sigma
,\tau $) are built up from constants and literals $v_{i},$ $\lnot v_{i}$ ($%
i\in \left[ n\right] $) by positive operations $\vee $ and $\wedge $. For
brevity we also stipulate $1 \vee \varphi =\varphi \vee 1 :=1 $, $0 \wedge
\varphi =\varphi \wedge 0 :=0 $ and $1 \wedge \varphi =\varphi \wedge 1 =0
\vee \varphi =0 \vee \varphi =\varphi \vee 0 \!:=\!\varphi $. Let $cs\left(
\varphi \right) $ denote structural complexity (i.e.\negthinspace\
circuit\negthinspace\ size) of $\varphi $. \footnote{{\small {\footnotesize %
More precisely, $\mathrm{cs}\left( \varphi \right) $ is the total number of
pairwise distinct subterms of (including) $\varphi $. }}} \newline
De\thinspace Morgan rules for negation provide length-preserving
interpretation of full Boolean algebra.

\item  We define by recursion on $cs\left( \varphi \right) $ two assignments 
\begin{eqnarray*}
\mathcal{A} &\ni &\varphi \hookrightarrow \mathrm{DN}\left( \varphi \right)
\in \left\{ 1\right\} \cup \wp \mathcal{D}\text{\ and } \\
\mathcal{A} &\ni &\varphi \hookrightarrow \mathrm{AP}\left( \varphi \right)
\in \left\{ 1\right\} \cup \wp _{L}\mathcal{D}^{\ell }
\end{eqnarray*}
that\ represent DNFs and corresponding approximations of $\varphi $,
respectively.

\begin{enumerate}
\item  $\mathrm{DN}\left( 1\right) =\mathrm{AP}\left( 1\right) :=1,\ \mathrm{%
DN}\left( 0\right) =\newline
\mathrm{AP}\left( 0\right) :=\emptyset .$

\item  $\mathrm{DN}\left( v_{i}\right) =\mathrm{AP}\left( v_{i}\right)
:=\left\{ \left\langle \left\{ \pi \left( i\right) \right\} ,\,\emptyset
\right\rangle \right\} .$

\item  $\mathrm{DN}\left( \lnot v_{i}\right) =\mathrm{AP}\left( \lnot
v_{i}\right) :=\left\{ \left\langle \emptyset \,,\left\{ \pi \left( i\right)
\right\} \right\rangle \right\} .$

\item  $\mathrm{DN}\left( \sigma \vee \delta \right) :=\mathrm{DN}\left(
\sigma \right) \cup \mathrm{DN}\left( \delta \right) ,\newline
\mathrm{AP}\left( \sigma \vee \delta \right) :=\mathrm{AP}\left( \sigma
\right) \sqcup \mathrm{AP}\left( \delta \right) .$

\item  $\mathrm{DN}\left( \sigma \wedge \delta \right) :=\mathrm{DN}\left(
\sigma \right) \odot \mathrm{DN}\left( \tau \right) ,\newline
\mathrm{AP}\left( \sigma \wedge \delta \right) :=\mathrm{AP}\left( \sigma
\right) \sqcap \mathrm{AP}\left( \delta \right) .$
\end{enumerate}

Thus for any $\varphi =\overset{r}{\underset{i=1}{\bigvee }}\varphi _{i} $
and $\psi =\underset{i\in I}{\bigwedge }v_{i}\wedge \underset{j\in J}{%
\bigwedge }\lnot v_{j}$, where $I\cap J=\emptyset $, we have $\mathrm{DN}%
\left( \varphi \right) =$ $\overset{r}{\underset{i=1}{\bigcup }}\mathrm{DN}%
\left( \varphi _{i}\right) $ and $\mathrm{DN}\left( \psi \right) =\left\{
\left\langle G,H\right\rangle \right\} $ for $G:=\left\{ \pi \left( i\right)
:\imath \in I\right\} $ and $H:=\left\{ \pi \left( j\right) :\jmath \in
j\right\} $. By the same token, $\wp \mathcal{D}=\left\{ \mathrm{DN}\left(
\varphi \right) :\varphi \in \mathcal{A}\right\} $.

\item  For any $\varphi \in \mathcal{A}$ we define total deviations $%
\partial ^{\text{\textsc{p}}}\left( \varphi \right) $ and $\partial ^{\text{%
\textsc{n}}}\left( \varphi \right) $ as follows, where $\mathrm{AC}^{\text{%
\textsc{p}}} \left( 1\right) :=\mathrm{POS}$ and $\mathrm{AC}^{\text{\textsc{%
n}}}\left( 1 \right) :=\mathrm{NEG}$, while $\mathrm{AC}^{\text{\textsc{p}}%
}\left( \mathrm{DN}\left( \varphi \right) \right) $ and $\mathrm{AC}^{\text{%
\textsc{n}}}\left( \mathrm{DN}\left( \varphi \right)\right) $ abbreviate $%
\mathrm{AC}^{\text{\textsc{p}}}\left( \varphi \right) $ and $\mathrm{AC}^{%
\text{\textsc{n}}}\left( \varphi \right) $, respectively.

\begin{enumerate}
\item  $\partial ^{\text{\textsc{p}}}\left( \varphi \right) :=\mathrm{AC}^{%
\text{\textsc{p}}}\left( \varphi \right) \setminus \mathrm{AC}^{\text{%
\textsc{p}}}\left( \mathrm{AP}\left( \varphi \right) \right) .$

\item  $\partial ^{\text{\textsc{n}}}\left( \varphi \right) :=\mathrm{AC}^{%
\text{\textsc{n}}}\left( \mathrm{AP}\left( \varphi \right) \right) \setminus 
\mathrm{AC}^{\text{\textsc{n}}}\left( \varphi \right) .$
\end{enumerate}
\end{itemize}

\begin{lemma}
For any $\sigma ,\delta \in\mathcal{A}$ the following holds.

\begin{enumerate}
\item  $\partial ^{\text{\textsc{p}}}\left( \sigma \vee \delta \right)
\subseteq \partial ^{\text{\textsc{p}}}\left( \sigma \right) \cup \partial ^{%
\text{\textsc{p}}}\left( \delta \right) \cup \partial _{\sqcup }^{\text{%
\textsc{p}}}\left( \mathrm{AP}\left( \sigma \right) ,\mathrm{AP}\left(
\delta \right) \right) .$

\item  $\partial ^{\text{\textsc{p}}}\left( \sigma \wedge \delta \right)
\subseteq \partial ^{\text{\textsc{p}}}\left( \sigma \right) \cup \partial ^{%
\text{\textsc{p}}}\left( \delta \right) \cup \partial _{\sqcap }^{\text{%
\textsc{p}}}\left( \mathrm{AP}\left( \sigma \right) ,\mathrm{AP}\left(
\delta \right) \right) .$

\item  $\partial ^{\text{\textsc{n}}}\left( \sigma \vee \delta \right)
\subseteq \partial ^{\text{\textsc{n}}}\left( \sigma \right) \cup \partial ^{%
\text{\textsc{n}}}\left( \delta \right) \cup \partial _{\sqcup }^{\text{%
\textsc{n}}}\left( \mathrm{AP}\left( \sigma \right) ,\mathrm{AP}\left(
\delta \right) \right) .$

\item  $\partial ^{\text{\textsc{n}}}\left( \sigma \wedge \delta \right)
\subseteq \partial ^{\text{\textsc{n}}}\left( \sigma \right) \cup \partial ^{%
\text{\textsc{n}}}\left( \delta \right) \cup \partial _{\sqcap }^{\text{%
\textsc{n}}}\left( \mathrm{AP}\left( \sigma \right) ,\mathrm{AP}\left(
\delta \right) \right) .$
\end{enumerate}
\end{lemma}

\begin{proof}
Straightforward via boolean inclusion\newline
$A\setminus B\subseteq \left( A\setminus C\right) \cup \left( C\setminus
B\right) $ (cf. Appendix B).
\end{proof}

\begin{lemma}
For any $\varphi \in\mathcal{A}$ the following conditions hold.

\begin{enumerate}
\item  $\left| \partial ^{\text{\textsc{p}}}\left( \varphi \right) \right|
<cs\left( \varphi \right) \cdot L^{2}\left( \QATOP{m-\ell -1}{k-\ell -1}%
\right) .$

\item  $\left| \partial ^{\text{\textsc{n}}}\left( \varphi \right) \right|
^{\ast }\leq cs\left( \varphi \right) \cdot 2^{-p}L^{2}\left| \mathcal{F}%
\right| $.

\item  If $\mathrm{AC}^{\text{\textsc{p\negthinspace }}}\left( \mathrm{AP\!}%
\left( \varphi \right) \right) \!\neq \!\emptyset \!$ then\negthinspace\, $%
\left| \mathrm{AC}^{\text{\textsc{n}}}\!\left( \mathrm{AP\!}\left( \varphi
\right) \right) \right| ^{\ast }\!\geq \!\dfrac{1}{2}\left| \mathcal{F}%
\right| $.
\end{enumerate}
\end{lemma}

\begin{proof}
1--2 follows from Lemmata 10, 11 by induction on $cs\left( \varphi \right) $.

3: $\mathrm{AC}^{\text{\textsc{p}}}\left( \mathrm{AP\!}\left( \varphi
\right) \right) \neq \!\emptyset $ implies $\mathrm{AP}\!\left( \varphi
\right) \neq \!\emptyset $, so there is at least one $D\in \mathrm{AP}%
\!\left( \varphi \right) $, $\left\| \text{\textsc{v}}\left( D\right)
\right\| \leq \ell $. Now by Lemma 9, $\left| \mathrm{AC}^{\text{\textsc{n}}%
}\left( \mathrm{AP}\left( \varphi \right) \right) \right| ^{\ast }\geq
\left| \mathcal{R}_{\subseteq }\left( D\right) \right| \geq \dfrac{1}{2}%
\left| \mathcal{F}\right| $, as $\mathrm{AC}^{\text{\textsc{n}}}\left( 
\mathrm{AP}\left( \varphi \right) \right) ^{\ast }\supseteq R_{\subseteq
}\left( D\right) $.
\end{proof}

\begin{itemize}
\item  Final assumptions$.$ Assuming $m\gg 0$ we let
\end{itemize}

$\fbox{$m=k^{4},\ k=2\ell ^{2},\ p=\ell \log _{2}\!m,\ L=\left(
p\!-\!1\right) ^{\ell }\ell !$}$

\begin{theorem}
\ Suppose that $\mathrm{AC}^{\text{\textsc{p}}}\left( \varphi \right) =%
\mathrm{POS}$ and $\mathrm{AC}^{\text{\textsc{n}}}\left( \varphi \right)
=\emptyset $ both hold for a given $\varphi \in \mathcal{A}$. Then for
sufficiently large $m$, $cs\left( \varphi \right) >m^{\frac{1}{5}m^{\frac{1}{%
8}}}$.
\end{theorem}

\begin{proof}
Consider two cases (cf. Appendix C).

1: Assume $\mathrm{AC}^{\mathrm{p}}\left( \mathrm{AP}\left( \varphi \right)
\right) =\emptyset $. By $\mathrm{AC}^{\text{\textsc{p}}}\left( \varphi
\right) =\mathrm{POS}$ we have $\partial ^{\text{\textsc{p}}}\left( \varphi
\right) =\mathrm{AC}^{\text{\textsc{p}}}\left( \varphi \right) \setminus 
\mathrm{AC}^{\text{\textsc{p}}}\left( \mathrm{AP}\left( \varphi \right)
\right) =\mathrm{POS}$. Hence by Lemma 13 (1),\newline
$\qquad cs\left( \varphi \right) \cdot \left( \QATOP{m-\ell -1}{k-\ell -1}%
\right) L^{2}\geq \left| \partial ^{\text{\textsc{p}}}\left( \varphi \right)
\right| =\left| \mathrm{POS}\right| =\left( \QATOP{m}{k}\right) .$\newline
Hence $cs\left( \varphi \right) \geq \left( \QATOP{m}{k}\right) \!\left( 
\QATOP{m-\ell -1}{k-\ell -1}\right) ^{-1}\!L^{-2}\!>\!\left( \frac{m-\ell }{k%
}\right) ^{\ell }\!L^{-2}\newline
>m^{\frac{1}{5}m^{\frac{1}{8}}}$.

2: Otherwise, assume $\mathrm{AC}^{\text{\textsc{p}}}\left( \mathrm{AP}%
\left( \varphi \right) \right) \neq \emptyset $. So $\mathrm{AC}^{\text{%
\textsc{n}}}\left( \varphi \right) =\emptyset $ implies $\partial ^{\text{%
\textsc{n}}}\left( \varphi \right) =\mathrm{AC}^{\text{\textsc{n}}}\left( 
\mathrm{AP}\left( \varphi \right) \right) \setminus \mathrm{AC}^{\text{%
\textsc{n}}}\left( \varphi \right) \newline
=\mathrm{AC}^{\text{\textsc{n}}}\left( \mathrm{AP}\left( \varphi \right)
\right) $. Hence $cs\left( \varphi \right) \cdot 2^{-p}L^{2}\left| \mathcal{F%
}\right| \geq \left| \partial ^{\text{\textsc{n}}}\left( \varphi \right)
\right| ^{\ast }\newline
\geq \dfrac{1}{2}\left| \mathcal{F}\right| $ by Lemma 13 (2, 3) and then $%
cs\left( \varphi \right) \geq 2^{p-1}L^{-2}>m^{\frac{1}{2}m^{\frac{1}{8}%
}}>m^{\frac{1}{5}m^{\frac{1}{8}}}$.
\end{proof}

\subsubsection{\protect\large Semantics}

\begin{definition}
Consider variable assignments 
\begin{equation*}
\mathrm{VA}=\left\{ \varepsilon :\left[ n\right] \rightarrow \left\{
0,1\right\} \right\} \text{.}
\end{equation*}
For any $i \in \left[ n\right] $, literals $v_{i}$, $\lnot v_{i}$, formulas $%
\varphi _{1},\cdots ,\varphi _{r}\in \mathcal{A}$ and \thinspace $%
\varepsilon \in \mathrm{VA}$, Boolean values $\left\| -\right\|
_{\varepsilon }\in \left\{ 0,1\right\} $ arise as follows.

\begin{enumerate}
\item  $\left\| 1\right\| _{\varepsilon }:=1,\ \left\| 0\right\|
_{\varepsilon }=0$.

\item  $\left\| v_{i}\right\| _{\varepsilon }:=\varepsilon \left( i\right) $.

\item  $\left\| \lnot v_{i}\right\| _{\varepsilon }:=1-\left\| v_{i}\right\|
_{\varepsilon }=1-\varepsilon \left( i\right) $.

\item  $\left\| \varphi _{1},\vee \cdots \vee \varphi _{r}\right\|
_{\varepsilon }:=\max \left\{ \left\| \varphi _{1}\right\| _{\varepsilon
},\cdots ,\left\| \varphi _{r}\right\| _{\varepsilon }\right\} $.

\item  $\left\| \varphi _{1}\wedge \cdots \wedge \varphi _{r}\right\|
_{\varepsilon }:=\min \left\{ \left\| \varphi _{1}\right\| _{\varepsilon
},\cdots ,\left\| \varphi _{r}\right\| _{\varepsilon }\right\} $.
\end{enumerate}

Furthermore, for any $D\in \mathcal{D}$, $\mathcal{X}\subseteq \mathcal{D}$
we define formulas $F\left( D\right) $, $F\left( \mathcal{X}\right) $ and
Boolean values $\left\| D\right\| _{\varepsilon }$, $\left\| \mathcal{X}%
\right\| _{\varepsilon }$:

\begin{enumerate}
\item  $F\left( D\right) :=\underset{\pi \left( i\right) \in D^{+}}{%
\bigwedge }v_{i}\wedge \underset{\pi \left( j\right) \in D^{-}}{\bigwedge }%
\lnot v_{j}$ and \newline
$F\left( \mathcal{X}\right) :=\underset{D\in \mathcal{X}}{\bigvee }F\left(
D\right) $.

\item  $\left\| D\right\| _{\varepsilon }:=\left\| \mathrm{F}\left( D\right)
\right\| _{\varepsilon }$ and $\left\| \mathcal{X}\right\| _{\varepsilon
}:=\left\| \mathrm{F}\left( \mathcal{X}\right) \right\| _{\varepsilon }$.
\end{enumerate}
\end{definition}

\begin{lemma}
$\left\| \varphi \right\| _{\varepsilon }=\left\| \mathrm{DN}\left( \varphi
\right) \right\| _{\varepsilon }$ holds for any $\varphi \in A$ and $%
\varepsilon \in \mathrm{VA}$.
\end{lemma}

\begin{proof}
We argue by induction on $cs\left( \varphi \right) $. Consider induction
step $\varphi =\sigma \wedge \delta $ where $\mathrm{DN}\left( \sigma
\right) ,\mathrm{DN}\left( \delta \right) \neq \emptyset $. So $\mathrm{DN}%
\left( \varphi \right) =\mathrm{DN}\left( \sigma \right) \odot \mathrm{DN}%
\left( \delta \right) =\left\{ D\Cup \!E:\left\langle D,E\right\rangle \in 
\mathrm{DN}\left( \sigma \right) \times \mathrm{DN}\left( \delta \right)
\right\} $, which yields

$\left\| \mathrm{DN}\left( \varphi \right) \right\| _{\varepsilon }=\bigvee
\left\{ \left\| D\Cup \!E\right\| _{\varepsilon }:\left\langle
D,E\right\rangle \in \mathrm{DN}\left( \sigma \right) \times \mathrm{DN\!}%
\left( \delta \right) \right\} $\newline
$=\bigvee \left\{ \left\| \left\langle D^{+}\cup \!E^{+},D^{-}\cup
E^{-}\right\rangle \right\| _{\varepsilon }:D\in \mathrm{DN}\left( \sigma
\right) \, \&\, E\in \mathrm{DN}\left( \delta \right) \right\} $.\newline
(We omit possible occurrences of $D\Cup E=\emptyset $ for \newline
$\left( D^{+}\!\cup \!E^{+}\right) \!\cap \!\left( D^{-}\!\cup
\!E^{-}\right) \!\neq \!\emptyset $, since they imply $\left\| \left\langle
D^{+}\cup \!E^{+},D^{-}\cup E^{-}\right\rangle \right\| _{\varepsilon }=0$.)
So by the induction hypothesis we get $\smallskip $

$\left\| \mathrm{DN}\left( \varphi \right) \right\| _{\varepsilon }=1
\smallskip $

$\Leftrightarrow \smallskip \left( \exists D\in \mathrm{DN}\left( \sigma
\right) \right) \left( \exists E\in \mathrm{DN}\left( \delta \right) \right) 
\newline
\qquad \left( 
\begin{array}{l}
\left( \forall \pi \left( i\right) \in D^{+}\cup E^{+}\right) \left\|
v_{i}\right\| _{\varepsilon }=1\ \& \\ 
\left( \forall \pi \left( j\right) \in D^{-}\cup E^{-}\right) \left\| \lnot
v_{j}\right\| _{\varepsilon }=1
\end{array}
\right) $

$\Leftrightarrow \smallskip \left( \exists D\in \mathrm{DN}\left( \sigma
\right) \right) \left( 
\begin{array}{l}
\left( \forall \pi \left( i\right) \in D^{+}\right) \left\| v_{i}\right\|
_{\varepsilon }=1\,\&\, \\ 
\left( \forall \pi \left( j\right) \in D^{-}\right) \left\| \lnot
v_{j}\right\| _{\varepsilon }=1
\end{array}
\right) \ \&$

$\quad \ \left( \exists E\in \mathrm{DN\!}\left( \delta \right) \right)
\left( 
\begin{array}{l}
\left( \forall \pi \left( i\right) \in E^{+}\right) \left\| v_{i}\right\|
_{\varepsilon }=1\,\&\, \\ 
\left( \forall \pi \left( j\right) \in E^{-}\right) \left\| \lnot
v_{j}\right\| _{\varepsilon }=1
\end{array}
\right) \smallskip $

$\Leftrightarrow \smallskip \left\| \mathrm{DN}\left( \sigma \right)
\right\| _{\varepsilon }=1=\left\| \mathrm{DN}\left( \delta \right) \right\|
_{\varepsilon }$

$\Leftrightarrow \left\| \sigma \right\| _{\widehat{\varepsilon }}=1=\left\|
\delta \right\| _{\varepsilon }$,\smallskip \newline
which yields $\smallskip \left\| \varphi \right\| _{\varepsilon }=\left\|
\sigma \wedge \delta \right\| _{\varepsilon }=\left\| \mathrm{DN}\left(
\varphi \right) \right\| _{\varepsilon }$.

Basis of induction and case $\varphi =\sigma \vee \tau $ are trivial.
\end{proof}

\begin{definition}
For any $\mathcal{X,Y}\subseteq \mathcal{D}$ we let\newline
$\mathrm{VA}_{0}:=\left\{ \varepsilon _{0}\in \mathrm{VA}:\left| \varepsilon
^{-1}\left( 1\right) \right| \leq \binom{k}{2}\right\} $ \newline
and define two equivalences $\sim $ and $\sim _{0}$ : \newline
\fbox{$
\begin{array}{l}
\mathcal{X}\sim \mathcal{Y:\Leftrightarrow }\left( \forall \varepsilon \in 
\mathrm{VA}\right) \left\| \mathcal{X}\right\| _{\varepsilon }=\left\| 
\mathcal{Y}\right\| _{\varepsilon }, \\ 
\mathcal{X}\sim _{0}\mathcal{Y:\Leftrightarrow }\left( \forall \varepsilon
\in \mathrm{VA}_{0}\right) \left\| \mathcal{X}\right\| _{\varepsilon
}=\left\| \mathcal{Y}\right\| _{\varepsilon }.
\end{array}
$}\smallskip

Obviously $\sim $ in stronger than $\sim _{0}$.
\end{definition}

\begin{lemma}
\item  Suppose $\varphi \in \!\mathcal{A}$ satisfies $\mathrm{DN}\left(
\varphi \right) \sim _{0}\mathrm{CLIQ}_{2}$. Then $\mathrm{AC}^{\text{%
\textsc{p}}}\left( \varphi \right) \!=\!\mathrm{POS}$ and $\mathrm{AC}^{%
\text{\textsc{n}}}\left( \varphi \right) \!=\!\emptyset $. So by Theorem 14, 
$cs\left( \varphi \right) >m^{\frac{1}{5}m^{\frac{1}{8}}}$ for sufficiently
large $m$. By Lemma 16, the latter holds for any $\varphi \sim _{0}\mathrm{%
CLIQ}_{2}$.
\end{lemma}

\begin{proof}
Suppose $\mathrm{DN}\left( \varphi \right) \sim _{0}\mathrm{CLIQ}_{2}$. We
show that $\mathrm{AC}^{\text{\textsc{p}}}\left( \varphi \right) \!=\!%
\mathrm{POS}$. Let $K=\left\{ \pi \left( i\right) :i\in S\right\} \in 
\mathrm{POS}$, which yields $\left\| \mathrm{CLIQ}_{2}\right\| _{\varepsilon
}=1$ for $\varepsilon \in \mathrm{VA}_{0}$ with 
\begin{equation*}
\varepsilon \left( i\right) \!:=\!\!\left\{ 
\begin{array}{lll}
1,\! & \!\text{if\negthinspace \negthinspace }\!\! & \!\!\!i\in S\text{,} \\ 
0,\! & \!\text{else,} & \!\!\!
\end{array}
\right. \ 
\end{equation*}
So $\left\| \mathrm{DN}\left( \varphi \right) \right\| _{\varepsilon }=1$
holds by the assumption, and hence there exists $D\in \mathrm{DN}\left(
\varphi \right) $ satisfying $\left\| D\right\| _{\varepsilon }=\!1$ for $%
D^{+}=\left\{ \pi \left( i\right) :i\in T\right\} $. But then for any $i\in
T $ we get $\varepsilon \left( i\right) =1$, which yields $T\subseteq S$ and
hence $D\subseteq ^{+}K$. So $\mathrm{POS}\subseteq \mathrm{AC}^{\text{%
\textsc{p}}}\left( \varphi \right) \subseteq \mathrm{POS}$. Thus $\mathrm{AC}%
^{\text{\textsc{p}}}\left( \varphi \right) =\mathrm{POS}$, as required.

Suppose there is a $C_{f}\in \mathrm{AC}^{\text{\textsc{n}}}\left( \varphi
\right) $, i.e. there exists $E\!\in \mathrm{DN}\left( \varphi \right) $
with $E\subseteq ^{+}C_{f}$ for $E^{+}\!=\!\left\{ \pi \left( a_{i}\right)
:i\in S\right\} $. Define $\varepsilon \in \mathrm{VA}_{0}$ as above. Then $%
\left\| E\right\| _{\varepsilon }=\!1$ and hence $\left\| \mathrm{DN}\left(
\varphi \right) \right\| _{\varepsilon }=1$. Hence $\left\| \mathrm{CLIQ}%
_{2}\right\| =1$ holds by the assumption, and therefore there exist $K\in 
\mathrm{POS}$ and $D\in \mathcal{D}$ such that\ $K\subseteq D^{+}$\ and $%
\left\| D\right\| _{\varepsilon }=1$. But arguing as above this would imply $%
D^{+}\subseteq E^{+}$ and hence $K\subseteq E^{+}\subseteq C_{f}$ , which
contradicts Lemma 1. Thus $\mathrm{AC}^{\text{\textsc{n}}}\left( \varphi
\right) =\emptyset $, as required.
\end{proof}

\subsection{\protect\large General Boolean case}

\begin{itemize}
\item  Let $\mathcal{B}$ denote full Boolean algebra with constants $1,0$,
operations $\vee ,\wedge ,\lnot $ and variables $v_{1},\cdots ,v_{n}$. Note
that $\mathcal{A\subset B}$.

\item  Arbitrary Boolean formulas $\varphi \in \mathcal{B}$ are convertible
to equivalent DMN formulas $\varphi ^{\ast }\in \mathcal{A}$ that arise by
applying as long as possible De\thinspace Morgan rewriting rules 1--4 :
\end{itemize}

\begin{enumerate}
\item  $\lnot \,1\hookrightarrow 0\,,\ \lnot \,0\hookrightarrow 1.$

\item  $\lnot \left( \sigma \vee \tau \right) \hookrightarrow \lnot \,\sigma
\wedge \lnot \,\tau .$

\item  $\lnot \left( \sigma \wedge \tau \right) \hookrightarrow \lnot \sigma
\vee \lnot \tau .$

\item  $\lnot \lnot \sigma \hookrightarrow \sigma .$
\end{enumerate}

It is a folklore that circuit size of $\varphi ^{\ast }$\ is at most that of 
$\varphi $. Namely, circuit structure of $\varphi ^{\ast }$ arises by
successively replacing $\lnot \,1$ by $0$, $\lnot \,0$ by $1$, gates $\vee $
and $\wedge $ occurring in the scope of $\lnot $ in $\varphi $ by
complementary gates $\wedge $ and $\vee $, respectively, while deleting
gates $\lnot $ in question and all double negations $\lnot \lnot $. Note
that $\lnot $ may occur in $\varphi ^{\ast }$\ only in literals $\lnot v_{i}$%
, if at all.

\begin{itemize}
\item  Semantics in $\mathcal{B}$ is defined as in $\mathcal{A}$ with
respect to variable assignments $\varepsilon \in \mathrm{VA}$ (cf.
Definition 15).
\end{itemize}

\begin{lemma}
For any $\varphi \in \mathcal{B}$ and $\varepsilon \in \mathrm{VA}$ we have $%
cs\left( \varphi ^{\ast }\right)\!\leq \!cs\left( \varphi \right) $ and 
\begin{equation*}
\left\| \varphi \right\| _{\varepsilon }=\left\| \varphi ^{\ast }\right\|
_{\varepsilon }=\left\| \mathrm{DN}\left( \varphi ^{\ast }\right) \right\|
_{\varepsilon }.
\end{equation*}
\end{lemma}

\begin{proof}
$\left\| \varphi \right\| _{\varepsilon }=\left\| \varphi ^{\ast }\right\|
_{\varepsilon }$ holds by trivial induction on $\left\| \varphi ^{\ast
}\right\| _{\varepsilon }=\left\| \mathrm{DN}\left( \varphi ^{\ast }\right)
\right\| _{\varepsilon }$, while $\left\| \varphi ^{\ast }\right\|
_{\varepsilon }=\left\| \mathrm{DN}\left( \varphi ^{\ast }\right) \right\|
_{\varepsilon }$ follows from Lemma 16.
\end{proof}

\begin{theorem}
Suppose that $\varphi \!\in \!\mathcal{B}$ provides a solution of $\mathrm{%
CLIQUE}_{m,k}$ in full Boolean logic and semantics involved. Then for
sufficiently large $m=k^{4}$, $cs\left( \varphi \right) $ is exponentially
large in $m$.
\end{theorem}

\begin{proof}
Without loss of generality assume that double graphs are represented by
arbitrary pairs of plain graphs $D=D\left( \varepsilon \right) =\left\langle
G\left( \varepsilon \right) ,G^{\prime }\right\rangle $ such that $G\left(
\varepsilon \right) =\left\{ \pi (i):i\in \left[ n\right] \,\&\,\,%
\varepsilon \left( i\right) =1\right\} $ and $G^{\prime }\subseteq \left\{
\pi (i):i\in \left[ n\right] \,\&\,\,\varepsilon \left( i\right) =0\right\} $%
, for any chosen $\varepsilon \in \mathrm{VA}$. For brevity we also write $%
D^{+}$ and $D^{-}$ for $G\left( \varepsilon \right) $ and $G^{\prime }$,
respectively, and denote by $\mathcal{D}$ the set of all $D=D\left(
\varepsilon \right) $ for $\varepsilon \in \mathrm{VA}$.

Consider Boolean circuits $C$ whose open sources are assigned with Boolean
constants and variables $v_{i}$ ($i\in \left[ n\right] $), and let $C\left(
\varepsilon \right) $ designate corresponding variable-free Boolean circuits
that are obtained by substituting $\varepsilon \left( i\right) $ for all $%
v_{i}$. Let $\left\| C\left( \varepsilon \right) \right\| \in
\left\{0,1\right\} $ denote the Boolean value of $C\left( \varepsilon
\right) $.

Now suppose that there exists a $C$ such that for every $\varepsilon \in 
\mathrm{VA}$, $C\left( \varepsilon \right) $ returns ``\textbf{true}'' iff $%
G\left( \varepsilon \right) $ contains a subgraph from $\mathrm{POS}$. In
our formalism this yields 
\begin{equation*}
\left( \forall \varepsilon \in \mathrm{VA}\right) \left( \left\| C\left(
\varepsilon \right) \right\| =1\Leftrightarrow CLIQ\left( \varepsilon
\right) \right) \qquad \left( \newline
3\right)
\end{equation*}
(cf. Introduction), provided that $C$ corresponds to Boolean formula $%
\varphi \in \mathcal{B}$.

Furthermore let $C^{\pm }$ denote a DMN circuit corresponding to DMN formula 
$\varphi ^{\ast }$. That is, $C^{\pm }$ has circuit structure of $\varphi
^{\ast }$ whose open sources are assigned with literals $v_{i}$ and/or $%
\lnot v_{i}$ occurring in $\varphi ^{\ast }$. For any $\varepsilon \in 
\mathrm{VA}$, $C^{\pm }\left( \varepsilon \right) $ will designate the
corresponding variable-free Boolean circuit obtained by substituting $%
\varepsilon \left( i\right) $ for all $v_{i}$, and let $\left\| C^{\pm
}\left( \varepsilon \right) \right\| \in \left\{ 0,1\right\} $ denote the
Boolean value of $C^{\pm }\left( \varepsilon \right) $. By Lemma 19, this
yields 
\begin{equation*}
\left( \forall \varepsilon \in \mathrm{VA}\right) \left\| C\left(
\varepsilon \right) \right\| =\left\| C^{\pm }\left( \varepsilon \right)
\right\| \qquad \qquad \qquad\ \ \ \ \ \left( \newline
2\right)
\end{equation*}
and 
\begin{equation*}
\left(\! \forall \varepsilon \in \mathrm{VA}\right) \left( 
\begin{array}{c}
\left\| C^{\pm }\left( \varepsilon \right) \right\| =1\Leftrightarrow \\ 
\left( \exists D\in \mathrm{DN}\left( C^{\pm }\right) \right) \left\|
C_{D}^{\pm }\left( \varepsilon \right) \right\| =1\smallskip
\end{array}
\!\right) \quad \left( 4\right)
\end{equation*}
which together with $\left( \newline
3\right) $ implies 
\begin{equation*}
\left( \forall \varepsilon \in \mathrm{VA}\right) \left( 
\begin{array}{c}
CLIQ\left( \varepsilon \right) \Leftrightarrow \\ 
\left( \exists D\in \mathrm{DN}\left( C^{\pm }\right) \right) \left\|
C_{D}^{\pm }\left( \varepsilon \right) \right\| =1\smallskip
\end{array}
\right) \quad \left( 5\right)
\end{equation*}
(cf. Introduction). Moreover, we prove another crucial equivalence 
\begin{equation*}
\left( \forall \varepsilon \in \mathrm{VA}\right) \left( CLIQ\left(
\varepsilon \right) \Leftrightarrow \left\| \mathrm{CLIQ}_{2}\right\|
_{\varepsilon }=1\smallskip \right) \qquad \qquad \left( 6\right)
\end{equation*}

$\vartriangleright $ $\smallskip \left\| \mathrm{CLIQ}_{2}\right\|
_{\varepsilon }=1\smallskip $

$\Leftrightarrow \left( \exists D\in \mathcal{D}\right) \left( \exists K\in 
\mathrm{POS}\right) \left( K\subseteq ^{+}D\,\&\left\| D\right\|
_{\varepsilon }=1\right) $

$\Leftrightarrow \left( \exists D\in \mathcal{D}\right) \left( \exists K\in 
\mathrm{POS}\right)$

$\qquad \qquad \qquad \qquad \left( \!\! 
\begin{array}{c}
K\subseteq D^{+}\& \\ 
\left\| \underset{\pi \left( i\right) \in D^{+}}{\bigwedge\!\! }v_{i}\wedge 
\underset{\pi \left( j\right) \in D^{-}}{\bigwedge }\!\!\lnot v_{j}\right\|
_{\varepsilon } \\ 
=1
\end{array}
\!\!\right) $

$\Leftrightarrow \left( \exists D\in \mathcal{D}\right) \left( \exists K\in 
\mathrm{POS}\right) $

$\qquad \qquad \qquad \qquad \left( 
\begin{array}{c}
K\subseteq D^{+}\,\& \\ 
\!\!\!\left. \! 
\begin{array}{c}
\left( \forall \pi \left( i\right) \in D^{+}\right) \varepsilon \left(
i\right) =1\,\& \\ 
\left( \forall \pi \left( j\right) \in D^{-}\right) \varepsilon \left(
j\right) =0
\end{array}
\right.
\end{array}
\!\!\!\!\right) $

$\smallskip\smallskip\Rightarrow \left( \exists D\in \mathcal{D}\right)
\left( \exists K\in \mathrm{POS}\right) \left( K\subseteq D^{+}\subseteq
G\left( \varepsilon \right) \right) $

$\Rightarrow CLIQ\left( \varepsilon \right) $

$\Rightarrow \left( \exists K\in \mathrm{POS}\right) \left( 
\begin{array}{c}
K\subseteq G\left( \varepsilon \right)\& \\ 
\left. 
\begin{array}{c}
\left( \forall \pi \left( i\right) \in G\left( \varepsilon \right) \right)
\varepsilon \left( i\right) =1\,\& \\ 
\left( \forall \pi \left( j\right) \in \emptyset \right) \varepsilon \left(
j\right) =0
\end{array}
\right.
\end{array}
\right) $\smallskip

$\Rightarrow \left( \exists K\!\in \!\mathrm{POS}\right) \left( 
\begin{array}{c}
\!K\!\subseteq ^{+}\!\!D\left( \varepsilon \right) :=\!\left\langle G\left(
\varepsilon \right) ,\emptyset \right\rangle\&, \\ 
\!\!\!\left. 
\begin{array}{c}
\left( \forall \pi \left( i\right) \in D\left( \varepsilon \right)
^{+}\right) \varepsilon \left( i\right) =1\,\& \\ 
\!\left( \forall \pi \left( j\right) \in D\left( \varepsilon \right)
^{-}\right) \varepsilon \left( j\right) =0
\end{array}
\!\!\right. \!
\end{array}
\right) $\smallskip

$\smallskip\smallskip\Leftrightarrow \left( \exists D\in \mathcal{D}\right)
\left( \exists K\in \mathrm{POS}\right) \left( K\subseteq ^{+}D\,\&\left\|
D\right\| _{\varepsilon }=1\right) $

$\smallskip \Leftrightarrow \left\| \mathrm{CLIQ}_{2}\right\| _{\varepsilon
}=1\vartriangleleft $

\smallskip which together with $\left( 5\right) $ implies 
\begin{equation*}
\left( \forall \varepsilon \in \mathrm{VA}\right) \left( 
\begin{array}{c}
\left\| \mathrm{CLIQ}_{2}\right\| _{\varepsilon }=1\Leftrightarrow \\ 
\left( \exists D\in \mathrm{DN}\left( C^{\pm }\right) \right) \left\|
C_{D}^{\pm }\left( \varepsilon \right) \right\| =1\smallskip
\end{array}
\right) \ \quad \left( 7\right)
\end{equation*}

It remains to observe that $\left( 7\right) $ is a circuit representation of
Lemma 18. Since $C$ and $C^{\pm }$ are respectively isomorphic to $\varphi $
and $\varphi ^{\ast } $, this completes the proof of Theorem.
\end{proof}

\begin{corollary}
It holds $\mathbf{NP}\nsubseteq \mathbf{P/poly}$. In particular $\mathbf{P}%
\neq \mathbf{NP}$.
\end{corollary}

\begin{proof}
Boolean circuit complexity is quadratic in derterministic time (cf. e.g. 
\cite{P}: Proposition 11.1, \cite{S}: Theorem 9.30). Hence the assertion
easily follows from Theorem 20 as CLIQUE$_{m,k}$ is a NP problem.
\end{proof}

\subsection{\protect\large Application}

Denote by $\mathcal{A}_{0}^{+}$ positive (monotone) subalgebra of $\mathcal{A%
}$ whose formulas are built up from variables and constants by positive
operations $\vee $ and $\wedge $. Thus CNF and/or DNF formulas $\varphi \in 
\mathcal{A} _{0}^{+}$ do not include negated variables.

\begin{theorem}
There is no polynomial time algorithm $f$ converting arbitrary CNF formulas $%
\varphi \in \mathcal{A}_{0}^{+}$ into equivalent DNF formulas $f\left(
\varphi \right) \in \mathcal{A}_{0}^{+}$.
\end{theorem}

\begin{proof}
Suppose that for every $\varepsilon :\left[ n\right] \rightarrow \left\{
0,1\right\},\newline
\left( \left\| \varphi \right\| _{\varepsilon }=1\Leftrightarrow \left\|
f\left( \varphi \right) \right\| _{\varepsilon }=1\Leftrightarrow \left\|
\lnot f\left( \varphi \right) \right\| _{\varepsilon }=0\right) $. Thus $%
\varphi \in \mathrm{SAT}\Leftrightarrow f\left( \varphi \right) \in \mathrm{%
SAT}\Leftrightarrow \lnot f\left( \varphi \right) \notin \mathrm{TAU}$.
Suppose that the size of $f\left( \varphi \right) $ is polynomial in that of 
$\varphi $. Note that $\lnot f\left( \varphi \right) \in \mathcal{B}$ is
equivalent to CNF formula $\left( \lnot f\left( \varphi \right) \right)
^{\ast }\in \mathcal{A}$ whose size is roughly the same as that of $f\left(
\varphi \right) $, and hence polynomial in the size of $\varphi $. \footnote{%
{\small {\footnotesize The difference between plain (linear) and circuit
length is inessential for CNF and/or DNF formulas under consideration.}}}
Also note that the validity problem $\left( \lnot f\left( \varphi \right)
\right) ^{\ast }\in ^{?}\mathrm{TAU}$ is solvable in polynomial time. Hence
so is the dual problem $\varphi \in ^{?}\mathrm{SAT}$, which by the NP
completeness of $\mathrm{SAT}$ yields $\mathbf{P}=\mathbf{NP}$, -- a
contradiction.
\end{proof}

\section{\protect\large Appendix A: On Lemma 9}

Let $\emptyset \neq G\in G^{\ell }$ and $R_{\subseteq }\left( G\right)
=\left\{ f\in \mathcal{F}:\,G\subseteq C_{f}\right\} $. To estimate $\left| 
\mathcal{R}_{\subseteq }\left( G\right) \right| $ we calculate the
probability that a coloring function $f\in \mathcal{F}$ is in $R_{\subseteq
}\left( G\right) $, i.e. every pair of nodes $x,y$ connected by an edge in $%
G $ is colored differently by $f\left( x\right) \neq f\left( y\right) <k$.
Therefore to color every next node in v$\left( G\right) $ we have to choose
an arbitrary color among those not previously used. This yields the
probability at least 
\begin{equation*}
\begin{array}{l}
\dfrac{k-1}{k-1}\times \dfrac{k-2}{k-1}\times \cdots \times \dfrac{%
k-1-\left| \text{\textsc{v\negthinspace }}\left( G\right) \right| }{k-1}%
>\smallskip \\ 
\left( \dfrac{k-1-\left| \text{\textsc{v\negthinspace }}\left( G\right)
\right| }{k-1}\right) ^{\left| \text{\textsc{v\negthinspace }}\left(
G\right) \right| } \geq \left( 1-\dfrac{\ell }{k-1}\right) ^{\ell
}>\smallskip \\ 
\left( 1-\dfrac{\ell }{k}\right) ^{\ell }>\smallskip \left( 1-\dfrac{1}{%
2\ell }\right) ^{\ell }\longrightarrow \frac{1}{\sqrt{e}}>\frac{1}{2}\text{
, }\smallskip \\ 
\text{as }k=2\ell ^{2}\longrightarrow \infty \text{ . }
\end{array}
\end{equation*}
Hence $\left| \mathcal{R}_{\subseteq }\left( G\right) \right| >\frac{1}{2}%
\left| \mathcal{F}\right| =\frac{1}{2}\left( k-1\right) ^{m}$, for
sufficiently large $k$. Now consider $R_{\nsubseteq }\left( G\right)
=\left\{ f\in \mathcal{F}:G\nsubseteq C_{f}\right\} =F\setminus R_{\subseteq
}\left( G\right) $ and make an obvious conclusion\newline
$\left| \mathcal{R}_{\nsubseteq }\left( G\right) \right| =\left| \mathcal{F}%
\right| -\left| {R}_{\subseteq }\left( G\right) \right| \leq \frac{1}{2}%
\left| \mathcal{F}\right| =\frac{1}{2}\left( k-1\right) ^{m}$. Consequently,
for any $D\in D^{\ell }$ we have \newline
$\left| \mathcal{R}_{\subseteq }\left( D\right) \right|\! \! =\! \! \left|
\left\{ f\in \mathcal{F}:D^{+}\subseteq C_{f}\right\} \right| \! \! >\! \!
\left| \mathcal{R}_{\subseteq }\left( D^{+}\right) \right| \! \! >\dfrac{1}{2%
}\left| \mathcal{F}\right| =\frac{1}{2}\left( k-1\right) ^{m}$, and hence $%
\left| \mathcal{R}_{\nsubseteq }\left( D\right) \right| \! \! \leq \! \! 
\dfrac{1}{2}\left| \mathcal{F}\right|\! \! =\! \! \frac{1}{2}\left(
k-1\right) ^{m}$.

Generally, for any $\mathcal{X}\subseteq \mathcal{F}$ we set $\mathcal{R}%
_{\subseteq }\left( \mathcal{X}:G\right) :=\left\{ f\in \mathcal{X}%
:G\subseteq C_{f}\right\} $ and $R_{\nsubseteq }\left( \mathcal{X}:G\right)
:=\left\{ f\in \mathcal{X}:G\nsubseteq C_{f}\right\} $. Then analogously $%
\left| \mathcal{R}_{\subseteq }\left( \mathcal{X}:G\right) \right| \geq 
\frac{1}{2}\left| \mathcal{X}\right| $ and $\left| \mathcal{R}_{\nsubseteq
}\left( \mathcal{X}:G\right) \right| \leq \frac{1}{2}\left| \mathcal{X}%
\right| $, provided that $\left| \mathcal{X}\left( x\right) \right| =k-1$
holds for any $x\in \,$v\negthinspace \thinspace $\left( G\right) $, where $%
\mathcal{X}\left( x\right) $ abbreviates $\left\{ f\left( x\right) :f\in 
\mathcal{X}\right\} $. Furthermore, for any $D\in \mathcal{D}^{\ell }$ we
set $\mathcal{R}_{\nsubseteq }\left( \mathcal{X}:D\right) :=\left\{ f\in 
\mathcal{X}:D^{+}\nsubseteq C\right\} $ and obtain $\left| \mathcal{R}%
_{\nsubseteq }\left( \mathcal{X}:D\right) \right| \leq \dfrac{1}{2}\left| 
\mathcal{X}\right| $, if $\left| \mathcal{X}^{+}\left( x\right) \right| =k-1$
for any $x\in \,$\textsc{v}\negthinspace \thinspace $\left( D^{+}\right) $.
Note that $\mathcal{R}_{\nsubseteq }\left( \mathcal{F}:D\right) =\mathcal{R}%
_{\nsubseteq }\left( D\right) $.

Consider any collection $D_{1},\cdots ,D_{q}\in \mathcal{D}^{\ell }$, $%
\left( \forall i\neq j\in \left[ q\right] \right) D_{i}^{+}\cap
D_{j}^{+}=\emptyset $. Then $\left| \overset{q}{\underset{i=1}{\bigcap }}%
\mathcal{R}_{\nsubseteq }\left( D_{i}\right) \right| \leq 2^{-q}\left| 
\mathcal{F}\right| $ will easily follow from 
\begin{equation*}
\left( \forall j\in \left[ q-1\right] \right) \!\left( \left| \overset{q}{%
\underset{i=j}{\bigcap }}\mathcal{R}_{\nsubseteq }\left( \mathcal{X\!}%
:\!D_{i}\right) \right| \leq 2^{-q}\!\left| \mathcal{X}\right| \right)
\qquad \left( \ast \right)
\end{equation*}
provided that $\mathcal{X}\subseteq \mathcal{F}$ satisfies $\left| \mathcal{X%
}^{+}\left( x_{i}\right) \right| =k-1$ for all $x_{i}\in \,$\textsc{v}%
\negthinspace \thinspace $\left( D_{i}^{+}\right) $, $i\in \left[ q\right] $%
. Now $\left( \ast \right) $ is proved as follows by induction on $q$.

\textbf{Basis} : $q=2$. Since $D_{1}^{+}\cap D_{2}^{+}=\emptyset $, for any $%
x_{1}\in \,$\textsc{v}\negthinspace \thinspace $\left( D_{i}^{+}\right) $, $%
x_{2}\in \,$\textsc{v}\negthinspace \thinspace $\left( D_{2}^{+}\right) $\
we have $\left| \mathcal{R}_{\nsubseteq }\left( \mathcal{X}:D_{1}\right)
\left( x_{2}\right) \right| =\left| \mathcal{X}\left( x_{2}\right) \right| $
and $\left| \mathcal{X}\left( x_{1}\right) \right| =\left| \mathcal{X}\left(
x_{2}\right) \right| =k-1$. This yields 
\begin{equation*}
\begin{array}{l}
\left| \mathcal{R}_{\nsubseteq }\!\left( \mathcal{X\!}\!:D_{1}\right) \!\cap
\!\mathcal{R}_{\nsubseteq }\!\left( \mathcal{X\!}:\!D_{2}\right) \right|
\!=\!\left| \mathcal{R}_{\nsubseteq }\!\left( \mathcal{R}_{\nsubseteq
}\!\left( \mathcal{X}\!:\!D_{1}\right) \!:\!D_{2}\right) \right| \\ 
\!\leq \!\dfrac{1}{2}\left| \mathcal{R}_{\nsubseteq }\!\left( \mathcal{X}%
\!:\!D_{1}\right) \right| \!\leq \!\dfrac{1}{4}\left| \mathcal{X}\right|
\end{array}
\end{equation*}

\textbf{Induction step}. By the same token we obtain 
\begin{equation*}
\begin{array}{l}
\ \ \left| \overset{q}{\underset{i=j}{\bigcap }}\mathcal{R}_{\nsubseteq
}\left( \mathcal{X\!}:\!D_{i}\right) \right| \!=\!\left| \overset{q-1}{%
\underset{i=j}{\bigcap }}\mathcal{R}_{\nsubseteq }\left( \mathcal{X\!}%
:\!D_{i}\right) \cap \mathcal{R}_{\nsubseteq }\left( \mathcal{X\!}%
:\!D_{q}\right) \right| \\ 
=\left| \mathcal{R}_{\nsubseteq }\left( \overset{q-1}{\underset{i=j}{\bigcap 
}}\mathcal{R}_{\nsubseteq }\left( \mathcal{X\!}:\!D_{i}\right) :D_{q}\right)
\right| \!\!\!\!\leq \!\!\!\!\dfrac{1}{2}\left| \overset{q-1}{\underset{i=j}{%
\bigcap }}\mathcal{R}_{\nsubseteq }\left( \mathcal{X\!}:\!D_{i}\right)
\right| \\ 
\leq 2^{-q}\left| \mathcal{X}\right| \text{.}
\end{array}
\end{equation*}

\section{\protect\large Appendix B: Proof of Lemma 12}

Use Lemma 3 and inclusion\smallskip\ $A\setminus B\subseteq \left(
A\setminus C\right) \cup \left( C\setminus B\right) \medskip $.

1. $\ \partial ^{\text{\textsc{p}}}\left( \sigma \vee \tau \right) =$

$\mathrm{AC}^{\text{\textsc{p}}}\left( \mathrm{DN}\left( \sigma \right)
\!\cup \!\mathrm{DN}\left( \tau \right) \right)\setminus\mathrm{AC}^{\text{%
\textsc{p}}}\left( \mathrm{AP}\left( \sigma \right) \!\sqcup \!\mathrm{AP}%
\left( \tau \right) \right) \smallskip$ $
\begin{array}{l}
\subseteq\!\mathrm{AC}^{\text{\textsc{p}}}\!\left( \mathrm{DN}\left( \sigma
\right) \!\cup \mathrm{DN}\left( \tau \right) \right) \setminus \!\mathrm{AC}%
^{\text{\textsc{p}}}\!\left( \mathrm{AP}\!\left( \sigma \right) \!\sqcup \!%
\mathrm{AC}^{\text{\textsc{p}}}\!\left( \mathrm{AP}\!\left( \tau \right)
\right) \right) \\ 
\cup\ \mathrm{AC}^{\text{\textsc{p}}}\!\left( \mathrm{AP}\left( \sigma
\right) \!\sqcup \!\mathrm{AC}^{\text{\textsc{p}}}\!\left( \mathrm{AP}\left(
\tau \right) \right) \right) \setminus \!\mathrm{AC}^{\text{\textsc{p}}%
}\!\left( \mathrm{AP}\left( \sigma \right) \!\sqcup \!\mathrm{AP}\left( \tau
\right) \right)
\end{array}
\smallskip $ $
\begin{array}{l}
=\mathrm{AC}^{\text{\textsc{p}}}\!\left( \mathrm{DN}\left( \sigma \right)
\!\cup \mathrm{AC}^{\text{\textsc{p}}}\left( \tau \right) \right) \setminus
\!\mathrm{AC}^{\text{\textsc{p}}}\!\left( \mathrm{AP}\!\left( \sigma \right)
\!\cup \!\mathrm{AC}^{\text{\textsc{p}}}\!\left( \mathrm{AP}\!\left( \tau
\right) \right) \right) \\ 
\cup \ \mathrm{AC}^{\text{\textsc{p}}}\!\left( \mathrm{AP}\left( \sigma
\right) \!\cup \!\mathrm{AC}^{\text{\textsc{p}}}\!\left( \mathrm{AP}\left(
\tau \right) \right) \right) \setminus \!\mathrm{AC}^{\text{\textsc{p}}%
}\!\left( \mathrm{AP}\left( \sigma \right) \!\sqcup \!\mathrm{AP}\left( \tau
\right) \right)
\end{array}
\smallskip $ $
\begin{array}{l}
\subseteq \!\mathrm{AC}^{\text{\textsc{p}}}\!\left( \sigma \right)
\!\setminus \!\mathrm{AC}^{\text{\textsc{p}}}\!\left( \mathrm{AP}\!\left(
\sigma \right) \!\right) \cup \mathrm{AC}^{\text{\textsc{p}}}\left( \tau
\right) \setminus \mathrm{AC}^{\text{\textsc{p}}}\left( \mathrm{AP}\left(
\tau \right) \right) \\ 
\cup \ \mathrm{AC}^{\text{\textsc{p}}}\!\left( \mathrm{AP}\!\left( \sigma
\right) \right) \cup \mathrm{AC}^{\text{\textsc{p}}}\!\left( \mathrm{AP}%
\!\left( \tau \right) \right) \setminus \!AC^{\text{\textsc{p}}}\!\left( 
\mathrm{AP}\!\left( \sigma \right) \sqcup \mathrm{AP}\!\left( \tau \right)
\right)
\end{array}
\smallskip $ $=\partial ^{\text{\textsc{p}}}\left( \sigma \right) \cup
\partial ^{\text{\textsc{p}}}\left( \tau \right) \cup \partial _{\sqcup }^{%
\text{\textsc{p}}}\left( \mathrm{AP}\left( \sigma \right) ,\mathrm{AR}\left(
\tau \right) \right) .\smallskip $

2. $\ \ \partial ^{\text{\textsc{p}}}\left( \sigma \wedge \tau \right) =$

$\mathrm{AC}^{\text{\textsc{p}}}\left( \mathrm{DN}\left( \sigma \right)
\odot \mathrm{DN}\left( \tau \right) \right) \setminus \mathrm{AC}^{\text{%
\textsc{p}}}\left( \mathrm{AP}\left( \sigma \right) \sqcap \mathrm{AP}\left(
\tau \right) \right) \smallskip $ $
\begin{array}{l}
\subseteq \!\mathrm{AC}^{\text{\textsc{p}}}\!\left( \mathrm{DN}\left( \sigma
\right) \!\odot \mathrm{DN}\left( \tau \right) \right) \setminus \!\mathrm{AC%
}^{\text{\textsc{p}}}\!\left( \mathrm{AP}\!\left( \sigma \right) \!\cap \!%
\mathrm{AC}^{\text{\textsc{p}}}\!\left( \mathrm{AP}\!\left( \tau \right)
\right) \right) \\ 
\cup \ \mathrm{AC}^{\text{\textsc{p}}}\!\left( \mathrm{AP}\left( \sigma
\right) \!\cap \!\mathrm{AC}^{\text{\textsc{p}}}\!\left( \mathrm{AP}\left(
\tau \right) \right) \right) \setminus \!\mathrm{AC}^{\text{\textsc{p}}%
}\!\left( \mathrm{AP}\left( \sigma \right) \sqcap \!\mathrm{AP}\left( \tau
\right) \right)
\end{array}
\smallskip $ $
\begin{array}{l}
=\mathrm{AC}^{\text{\textsc{p}}}\!\left( \sigma \right) \!\cap \mathrm{AC}^{%
\text{\textsc{p}}}\left( \tau \right) \setminus \!\mathrm{AC}^{\text{\textsc{%
p}}}\!\left( \mathrm{AP}\!\left( \sigma \right) \!\cap \!\mathrm{AC}^{\text{%
\textsc{p}}}\!\left( \mathrm{AP}\!\left( \tau \right) \right) \right) \\ 
\cup \ \mathrm{AC}^{\text{\textsc{p}}}\!\left( \mathrm{AP}\left( \sigma
\right) \!\cup \!\mathrm{AC}^{\text{\textsc{p}}}\!\left( \mathrm{AP}\left(
\tau \right) \right) \right) \setminus \!\mathrm{AC}^{\text{\textsc{p}}%
}\!\left( \mathrm{AP}\left( \sigma \right) \!\sqcap \mathrm{AP}\left( \tau
\right) \right)
\end{array}
\smallskip $ $
\begin{array}{l}
\subseteq \!\mathrm{AC}^{\text{\textsc{p}}}\!\left( \sigma \right)
\!\setminus \!\mathrm{AC}^{\text{\textsc{p}}}\!\left( \mathrm{AP}\!\left(
\sigma \right) \!\right) \cup \mathrm{AC}^{\text{\textsc{p}}}\left( \tau
\right) \setminus \mathrm{AC}^{\text{\textsc{p}}}\left( \mathrm{AP}\left(
\tau \right) \right) \\ 
\cup \ \ \partial _{\sqcap }^{\text{\textsc{p}}}\left( \mathrm{AP}\left(
\sigma \right) ,\mathrm{AR}\left( \tau \right) \right)
\end{array}
\smallskip $ $=\partial ^{\text{\textsc{p}}}\left( \sigma \right) \cup
\partial ^{\text{\textsc{p}}}\left( \tau \right) \cup \partial _{\sqcap }^{%
\text{\textsc{p}}}\left( \mathrm{AP}\left( \sigma \right) ,\mathrm{AR}\left(
\tau \right) \right) .\smallskip $

3. $\ \partial ^{\text{\textsc{n}}}\left( \sigma \vee \tau \right) =$

$\mathrm{AC}^{\text{\textsc{n}}}\left( \mathrm{AP}\left( \sigma \right)
\sqcup \mathrm{AP}\left( \tau \right) \right) \setminus \mathrm{AC}^{\text{%
\textsc{n}}}\left( \mathrm{DN}\left( \sigma \right) \cup \mathrm{DN}\left(
\tau \right) \right) \smallskip $ $
\begin{array}{l}
\subseteq \!\mathrm{AC}^{\text{\textsc{n}}}\left( \mathrm{AP}\left( \sigma
\right) \!\sqcup \!\mathrm{AP}\left( \tau \right) \right) \!\setminus \!%
\mathrm{AC}^{\text{\textsc{n}}}\left( \mathrm{AP}\left( \sigma \right)
\right) \!\cup \!\mathrm{AC}^{\text{\textsc{n}}}\left( \mathrm{AP}\left(
\tau \right) \right) \\ 
\cup \ \mathrm{AC}^{\text{\textsc{n}}}\!\left( \mathrm{AP}\!\left( \sigma
\right) \right) \!\cup \!\mathrm{AC}^{\text{\textsc{n}}}\!\left( \mathrm{AP}%
\!\left( \tau \right) \right) \setminus \mathrm{AC}^{\text{\textsc{n}}%
}\!\left( \mathrm{DN}\left( \sigma \right) \!\cup \!\mathrm{DN}\left( \tau
\right) \right)
\end{array}
\smallskip $ $
\begin{array}{l}
=\mathrm{AC}^{\text{\textsc{n}}}\!\left( \mathrm{AP}\!\left( \sigma \right)
\sqcup \mathrm{AP}\!\left( \tau \right) \right) \!\setminus \!\mathrm{AC}^{%
\text{\textsc{n}}}\!\left( \mathrm{AP}\!\left( \sigma \right) \right) \cup 
\mathrm{AC}^{\text{\textsc{n}}}\!\left( \mathrm{AP}\!\left( \tau \right)
\right) \\ 
\cup \ \mathrm{AC}^{\text{\textsc{n}}}\left( \mathrm{AP}\left( \sigma
\right) \right) \cap \mathrm{AC}^{\text{\textsc{n}}}\left( \mathrm{AP}\left(
\tau \right) \right) \setminus \left[ \mathrm{AC}^{\text{\textsc{n}}}\left(
\sigma \right) \cup \mathrm{AC}^{\text{\textsc{n}}}\left( \tau \right) 
\right]
\end{array}
\smallskip $ $
\begin{array}{l}
\subseteq \!\partial _{\sqcup }^{\text{\textsc{n}}}\left( \mathrm{AP}\left(
\sigma \right) ,\mathrm{AR}\left( \tau \right) \right) \cup AC^{\text{%
\textsc{n}}}\left( \mathrm{AP}\left( \sigma \right) \right) \setminus 
\mathrm{AC}^{\text{\textsc{n}}}\left( \sigma \right) \\ 
\cup \ \mathrm{AC}^{\text{\textsc{n}}}\left( \mathrm{AP}\left( \tau \right)
\right) \setminus \mathrm{AC}^{\text{\textsc{n}}}\left( \tau \right)
\end{array}
\smallskip $ $=\partial _{\sqcup }^{\text{\textsc{n}}}\left( \mathrm{AP}%
\left( \sigma \right) ,\mathrm{AR}\left( \tau \right) \right) \cup \partial
^{\text{\textsc{n}}}\left( \sigma \right) \cup \partial ^{\text{\textsc{n}}%
}\left( \tau \right) .\smallskip $

4. $\ \partial ^{\text{\textsc{n}}}\left( \sigma \wedge \tau \right) =$

$\mathrm{AC}^{\text{\textsc{n}}}\left( \mathrm{AP}\left( \sigma \right)
\sqcap \mathrm{AP}\left( \tau \right) \right) \setminus \mathrm{AC}^{\text{%
\textsc{n}}}\left( \mathrm{DN}\left( \sigma \right) \odot \mathrm{DN}\left(
\tau \right) \right) \smallskip $ $
\begin{array}{l}
\subseteq \!\mathrm{AC}^{\text{\textsc{n}}}\left( \mathrm{AP}\left( \sigma
\right) \!\sqcap \!\mathrm{AP}\left( \tau \right) \right) \!\setminus \!%
\mathrm{AC}^{\text{\textsc{n}}}\left( \mathrm{AP}\left( \sigma \right)
\right) \!\cap \!\mathrm{AC}^{\text{\textsc{n}}}\left( \mathrm{AP}\left(
\tau \right) \right) \\ 
\cup \ \mathrm{AC}^{\text{\textsc{n}}}\!\left( \mathrm{AP}\!\left( \sigma
\right) \right) \!\cap \mathrm{AC}^{\text{\textsc{n}}}\!\left( \mathrm{AP}%
\!\left( \tau \right) \right) \setminus \mathrm{AC}^{\text{\textsc{n}}%
}\!\left( \mathrm{DN}\left( \sigma \right) \!\cap \!\mathrm{DN}\left( \tau
\right) \right)
\end{array}
\smallskip $ $
\begin{array}{l}
=\mathrm{AC}^{\text{\textsc{n}}}\!\left( \mathrm{AP}\!\left( \sigma \right)
\sqcap \mathrm{AP}\!\left( \tau \right) \right) \!\setminus \!\mathrm{AC}^{%
\text{\textsc{n}}}\!\left( \mathrm{AP}\!\left( \sigma \right) \right) \cap 
\mathrm{AC}^{\text{\textsc{n}}}\!\left( \mathrm{AP}\!\left( \tau \right)
\right) \\ 
\cup \ \mathrm{AC}^{\text{\textsc{n}}}\left( \mathrm{AP}\left( \sigma
\right) \right) \cup \mathrm{AC}^{\text{\textsc{n}}}\left( \mathrm{AP}\left(
\tau \right) \right) \setminus \left[ \mathrm{AC}^{\text{\textsc{n}}}\left(
\sigma \right) \cap \mathrm{AC}^{\text{\textsc{n}}}\left( \tau \right) 
\right]
\end{array}
\smallskip $ $
\begin{array}{l}
\subseteq \!\partial _{\sqcap }^{\text{\textsc{n}}}\left( \mathrm{AP}\left(
\sigma \right) ,\mathrm{AR}\left( \tau \right) \right) \cup \mathrm{AC}^{%
\text{\textsc{n}}}\left( \mathrm{AP}\left( \sigma \right) \right) \setminus 
\mathrm{AC}^{\text{\textsc{n}}}\left( \sigma \right) \smallskip \\ 
\cup \ \mathrm{AC}^{\text{\textsc{n}}}\left( \mathrm{AP}\left( \tau \right)
\right) \setminus \mathrm{AC}^{\text{\textsc{n}}}\left( \tau \right)
\end{array}
\smallskip $ $=\partial _{\sqcap }^{\text{\textsc{n}}}\left( \mathrm{AP}%
\left( \sigma \right) ,\mathrm{AR}\left( \tau \right) \right) \cup \partial
^{\text{\textsc{n}}}\left( \sigma \right) \cup \partial ^{\text{\textsc{n}}%
}\left( \tau \right) .\smallskip .$

\section{\protect\large Appendix C: Basic (in)equalities}

We have $k=m^{\frac{1}{4}}=2\ell ^{2},\ p=\ell \log _{2}\!m,\ L=\left(
p\!-\!1\right) ^{\ell }\ell !$, where $m\gg 0.$ So $\ell =\frac{1}{\sqrt{2}}%
m^{\frac{1}{8}}$, and hence

\QTP{Body Math}
$\ell !\thickapprox \!\sqrt{2\pi \ell }\left( \dfrac{\ell }{e}\right) ^{\ell
}\!=\sqrt{\!\sqrt{2}\pi m^{\frac{1}{8}}}\!\left( \dfrac{m^{\frac{1}{8}}}{%
\sqrt{2}e}\right) ^{\!\frac{1}{\sqrt{2}}m^{\frac{1}{8}}}$

$<m^{\frac{1}{16}+\frac{1}{8\sqrt{2}}m^{\frac{1}{8}}}\!<m^{\frac{1}{11}m^{%
\frac{1}{8}}}$, $m\gg 0.$ \smallskip

So\ \fbox{$\ell !<\!m^{\frac{1}{11}m^{\frac{1}{8}}}$}, while\ \smallskip 
\fbox{$\log_{2}m<m^{\alpha }$}

for any chosen $\alpha >0$.\smallskip

Now $p=\ell \log _{2}m<m^{\frac{1}{11}+\alpha }<m^{\frac{1}{10}}$, and hence
\smallskip

\fbox{$\left( p-1\right) ^{\ell }<p^{\ell }<m^{\frac{1}{10\sqrt{2}}m^{\frac{1%
}{8}}}<m^{\frac{1}{14}m^{\frac{1}{8}}}$}, while\smallskip

\fbox{$2^{p}=m^{\ell }=m^{\frac{1}{\sqrt{2}}m^{\frac{1}{8}}}$}. \smallskip

Thus $L\!=\!\left( p\!-\!1\right) ^{\ell }\ell ! \!<\!m^{\frac{1}{14}m^{%
\frac{1}{8}}}\!m^{\frac{1}{11}m^{\frac{1}{8}}}\!<\!m^{\frac{3}{50}m^{\frac{1%
}{8}}}$,

and hence \fbox{$L^{2}<m^{\frac{3}{25}m^{\frac{1}{8}}}$}. Moreover

$\left( \dfrac{m-\ell }{k}\right) ^{\ell }=\left( \dfrac{m-\frac{1}{\sqrt{2}}%
m^{\frac{1}{8}}}{m^{\frac{1}{4}}}\right) ^{\frac{1}{\sqrt{2}}m^{\frac{1}{8}%
}}\!\!>m^{\frac{1}{2\sqrt{2}}m^{\frac{1}{8}}}$,

and hence:\smallskip

\fbox{$\left( \dfrac{m\!-\!\ell }{k}\right) ^{\ell }\!L^{-2}\!>\!\dfrac{m^{%
\frac{1}{2\sqrt{2}}m^{\frac{1}{8}}}}{m^{\frac{3}{25}m^{\frac{1}{8}}}}\!>\!m^{%
\frac{1}{5}m^{\frac{1}{8}}}$}\smallskip

\fbox{$2^{p-1}L^{-2}\!>\!\frac{1}{2}\dfrac{m^{\frac{1}{\sqrt{2}}m^{\frac{1}{8%
}}}}{m^{\frac{3}{25}m^{\frac{1}{8}}}}\!>\!m^{\frac{1}{2}m^{\frac{1}{8}%
}}>\!m^{\frac{1}{5}m^{\frac{1}{8}}}$}. 

\end{document}